\documentclass[a4paper,twocolumn,11pt,accepted=2025-04-27]{quantumarticle}
\pdfoutput=1
\usepackage[utf8]{inputenc}
\usepackage{babel}
\usepackage[T1]{fontenc}
\usepackage{tikz}
\usepackage{graphicx,xcolor}
\definecolor{bl}{RGB}{104, 103, 137}
\definecolor{ye}{RGB}{229, 226, 185}
\definecolor{gr}{RGB}{153, 133, 126}
\usepackage{hyperref}
\usepackage[caption=false]{subfig}
\usepackage[normalem]{ulem}
\usepackage{braket,dsfont,amsthm,amssymb,amsmath}
\newtheorem{theorem}{Observation}

\newtheorem{lemma}[theorem]{Lemma}

\newcommand{\ty}[1]{{\lceil{#1}\rceil}}

\newcommand{\siq}{\mathcal{S}_{\rm I}}
\newcommand{\scq}{\mathcal{S}_{\rm C}}
\newcommand{\tr}{{\rm Tr}}
\newcommand{\id}{\mathds{1}}
\newcommand{\ghz}{{\rm GHZ}}

\newcommand{\etr}{{\cal E}_{\rm tr}}
\newcommand{\betr}{\bar{\cal E}_{\rm tr}}

\newcommand{\ecost}{{\cal E}_{c}}
\newcommand{\ebu}{{\cal E}_{bu}}
\newcommand{\efr}{{\cal E}_{w}}
\newcommand{\ecom}{{\cal E}_{r}}

\begin{document}
\title{Quantum network-entanglement measures}
\author{Zhen-Peng Xu}
\affiliation{School of Physics and Optoelectronics Engineering, Anhui University, 230601 Hefei, China}
\author{Julio I. de Vicente}
\affiliation{Departamento de Matem\'{a}ticas, Universidad Carlos III de Madrid, E-28911, Legan\'{e}s (Madrid), Spain}
\affiliation{Instituto de Ciencias Matem\'aticas (ICMAT), E-28049, Madrid, Spain}
\email{jdvicent@math.uc3m.es}
\author{Liang-Liang Sun}
\affiliation{Department of Modern Physics and National Laboratory for Physical Sciences at Microscale, University of Science and Technology of China, Hefei, Anhui 230026, China}
\author{Sixia Yu}
\email{yusixia@ustc.edu.cn}
\affiliation{Department of Modern Physics and National Laboratory for Physical Sciences at Microscale, University of Science and Technology of China, Hefei, Anhui 230026, China}
\affiliation{Hefei National Laboratory, University of Science and Technology of China, Hefei 230088, China}

\maketitle

\begin{abstract}
  Quantum networks are of high interest nowadays and a quantum internet has been long envisioned. Network-entanglement adapts the notion of entanglement to the network scenario and network-entangled states are considered to be a resource to overcome the limitations of a given network structure. In this work, we introduce measures of quantum network-entanglement that are well-defined within the general framework of quantum resource theories, which at the same time have a clear operational interpretation characterizing the extra resources necessary to prepare a targeted quantum state within a given network. In particular, we define the network communication cost and the network round complexity, which turn out to be intimately related to graph-theoretic parameters. We also provide methods to estimate these measures by introducing novel witnesses of network-entanglement. 
\end{abstract}

\section{Introduction}
In recent years, quantum networks have emerged as a highly promising platform for implementing quantum communication and computation tasks.
Quantum networks generalize basic primitives of quantum technologies such as quantum repeaters~\cite{briegel1998quantum,sangouard2011quantum} and they rely on the ability to generate and store entanglement at remote interconnected locations. Hence, there is a considerable experimental effort addressed at combining entanglement sources and quantum memories so as to meet the above demands (see e.g.\ \cite{liu2021heralded,lago2021telecom}). The applications of quantum networks include quantum key distribution~\cite{Murta2020,Hahn2020,siddhartha2021universal}, quantum sensing~\cite{yang2023quantumenhanced}, distributed quantum computation and even a quantum internet~\cite{kimble2008quantum,caleffi2018quantum,wehner2018quantum} in the long run. Furthermore, different tools for evaluating the performance of quantum networks have emerged gradually~\cite{azuma2021tools}. 
{Quantum correlations and, in particular, entanglement play a crucial role in the foundations of quantum mechanics and in the applications of quantum information theory~\cite{horodecki2009quantum,brunner2014bell,uola2009quantum}. T}here is a very active line of theoretical research directed at translating the different notions of quantum correlations to the network scenario, be it entanglement~\cite{navascues2020genuine,aberg2020semidefinite,kraft2021quantum}, nonlocality~\cite{renou2019genuine,contreras2021genuine,tavakoli2021bell,mao2022test} or steering~\cite{jones2021network}. In a similar way as quantum correlations with no classical analogue can be interpreted as a resource to overcome the limitations of a particular scenario, their network analogues should do the same when a prescribed network structure is imposed. 

In the most general scenario, a quantum network can be abstracted as a hypergraph, where each vertex represents a party, and each hyperedge stands for a quantum entangled source that distributes particles only to the parties associated with the corresponding vertices, see Fig.~\ref{fig:protocols} for an example. A state that cannot arise from a given network with arbitrary entangled states distributed according to its hyperedges and further processed with local operations and shared randomness (LOSR) is called network-entangled~\cite{navascues2020genuine} (or new genuine multipartite entangled~\cite{luo2021new} for special networks). While entanglement theory assumes classical communication for free, leading to the paradigm of local operations and classical communication (LOCC)~\cite{chitambar2014everything}, this is not the case here. First, this is because classical communication would render the theory trivial as all states can be prepared in a connected quantum network with LOCC operations by means of sequential teleportation~\cite{bennett1993teleporting,vaidman1994teleportation,bouwmeester1997experimental}. Second, LOCC protocols, which are structured in the form of a certain number rounds, require full control of distributed entangled states for considerably {long} periods. This places high demands on the local quantum memories, particularly for vertices that are far away from each other~\cite{liu2021heralded,lago2021telecom}.

As it turns out, many interesting states are network-entangled, like entangled symmetric states~\cite{hansenne2022symmetries} and all the qubit graph states corresponding to a connected graph with no less than $3$ vertices~\cite{wang2024quantum,makuta2023no}. Various techniques have emerged to distinguish between network states and network-entangled states, including semi-definite programming~\cite{navascues2020genuine}, covariance matrix~\cite{aberg2020semidefinite}, stabilizer analysis~\cite{hansenne2022symmetries} and network nonlocality inequalities~\cite{mao2022test}. Especially, the fidelity between GHZ states (graph states also)~\cite{navascues2020genuine,hansenne2022symmetries,makuta2023no,wang2024quantum} and network states has been estimated, which leads to fidelity-based witnesses for network-entangled states.

Thus, the above conditions make it possible to benchmark the devices in an experiment as being able to overcome the limitations of a given network structure. However, these tests tell us nothing about the robustness of the experimental set-up. {\color{black} To prepare a target network-entangled state, does it imply a significant overhead to the network setup or can this be done with only minimal additional resources?} To answer this question, rigorous means for the quantification of network-entanglement become indispensable, which can then unveil the potential of different network-entangled states for quantum distributed tasks. On the other hand, by constructing network-entanglement measures, one can address network design problems and find the network topology that minimizes the cost of preparing a given target state within the constraints dictated by a particular problem at hand, especially when the measures can be easily related to graph-theoretic properties of the corresponding hypergraph. In general, we hope that the development of network-entanglement measures might also help to identify tasks where quantum networks overcome classical networks in terms of communication complexity~\cite{kushilevitz1997communication,rao2020communication}, generalizing fundamental primitives such as superdense coding~\cite{bennett1992communication}.

While the basic condition for entanglement measures is LOCC monotonicity~\cite{horodecki2009quantum}, it should be clear from the above discussion that network-entaglement measures should obey LOSR monotonicity~\cite{navascues2020genuine}. In this work, we introduce the \textit{network-entanglement weight}, \textit{network communication cost} and the \textit{network round complexity} as measures of network-entanglement, which have clear operational interpretations as the minimal average rounds of classical communication within different classes of LOCC protocols that prepare the state from the network. Furthermore, we show that these quantifiers are well-defined from a resource-theoretic perspective by proving that they are LOSR monotonic, convex and subadditive. Interestingly, in combination with graph-theoretic parameters given by the underlying network, 
we show that the network-entanglement weight
gives both lower and upper bounds to the communication cost and round complexity. Moreover, 
these bounds are in general tight. Finally, the connection to the network-entanglement weight allows us to provide estimations of the introduced network-entanglement measures, for which we develop new witnesses.

\section{Measures of network-entanglement}
Quantum networks are characterized by hypergraphs $G=(V,E)$, where $V\subset\mathbb{N}$ is the set of vertices and $E$ is the set of hyperedges, whose elements are subsets of $V$ of cardinality at least $2$. In the particular case that all elements of $E$ are subsets of $V$ with cardinality equal to two, then we call $G$ a graph and the elements of $E$ edges. If we consider $n$-partite states, then $V=\{1,2,\ldots,n\}$ and every element of $V$ represents one of the parties. Every hyperedge $e\in E$ represents that all parties in $e$ are connected by a source and can share arbitrary quantum states among themselves. For a given network associated with the hypergraph $G$, any states of the form $\bigotimes_{e\in E}\sigma_e$, where $\sigma_e$ is an arbitrary multipartite quantum state shared by the parties in $e$, is considered to be freely available. We denote by $\scq(G)$ the closure of the set of states that can be prepared from those free states by LOSR. 
That is,
\begin{align}
  &\scq(G)\nonumber \\  =\ & {\operatorname{cl}}\left\{\rho\, \vert\, \rho \!=\! \sum\nolimits_{\lambda} p_{\lambda} [\otimes_{i\in V} {\cal C}^{(\lambda)}_i]\big([\otimes_{e\in E} \sigma_e]\big) \right\},
\end{align}
where ${\cal C}_i^{(\lambda)}$ is any local operation carried out by the $i$-th party, which may depend on a random variable $\lambda$ with probability distribution $\{p_{\lambda}\}$.
Any state out of this set is said to be network-entangled according to $G$ or, for short, $G$-entangled. In the particular case of hypergraphs $G$ where $E$ is given by all the subsets of $V$ with $k$ vertices, $G$-entangled states are also referred to as $k$-network-entangled. A proper measure ${\cal E}$ of $G$-entanglement should satisfy the following properties. 
\begin{enumerate}
    \item ${\cal E}(\rho) = 0$ for any state $\rho\in\scq(G)$;
    \item Monotonicity: ${\cal E}({\cal C}(\rho)) \le {\cal E}(\rho)$ for any LOSR operation ${\cal C}$.
\end{enumerate}
Other desired properties are \textit{convexity} and \textit{subadditivity}.
By applying the formalism of quantum resource theories~\cite{chitambar2019quantum} to the above paradigm, many network-entanglement measures can be constructed. 
Particularly, quantifiers based on the trace distance and fidelity have been often used in entanglement theory~\cite{horodecki2009quantum} and elsewhere. Their mathematical structure makes their estimation feasible and, in turn, they can be used to bound other measures~\cite{sun2024bounding}. More details are provided in Appendix (App.) X.
However, here we would like to focus on quantifiers that are operationally meaningful to the network scenario. Since classical communication is not free in this setting and it fuels any $G$ state to become $G$-entangled, a natural choice is to use the amount of this resource necessary to achieve this task through an LOCC protocol. Although the measures that we will construct turn out to be quite different, this approach has been similarly used to quantify quantum nonlocality~\cite{bacon2003bell}, where LOSR also gives rise to the free operations.

LOCC transformations are quantum channels that can be described by a sequence of rounds. In each round, based on the classical information produced at previous rounds, one party implements a local quantum channel whose output can be postselected (i.e.,\ a POVM corresponding to a Kraus decomposition of the channel) and broadcasts classical information to the other parties, who can then use it to apply local correction channels. The number of rounds of an LOCC protocol is then directly related to the number of uses of the broadcasting classical channel and gives rise to the notion of LOCC round complexity in entanglement theory~\cite{chitambar2011round,chitambar2017round}. In the network scenario, other possibilities appear naturally in order to measure the cost of an LOCC protocol. First, the network structure might further constrain which parties share classical channels. Second, different rounds might only involve subsets of parties that are far away in network-distance and can therefore be carried out simultaneously adding no extra time delay. Thus, given an LOCC protocol, we propose an alternative way to measure the complexity of implementing it in a network $G$ by decomposing it as a sequence of actions that we call \textit{steps}. In a given step, each party is identified either as a sender or as a receiver. Based on their knowledge of the classical information produced at previous steps, senders can implement any local quantum channel and then send classical information conditioned on the output of this channel but only to neighbouring receivers in $G$. The receivers can then apply local correction channels. 

Notice that, while in one round only one party can implement a local quantum channel, in one step several parties might do so. However, while the result of this channel can be broadcast to all other parties in a round, the spread of this information is more restricted in the latter as in each step this information can only travel to parties that share a hyperedge with the sender. Thus, certain LOCC protocols are more costly in terms of rounds and others in terms of steps (see Fig.\ 1 for a simple example).

\begin{figure}[th]
\subfloat[\label{subfig:b}]{%
  \includegraphics[width=0.42\columnwidth]{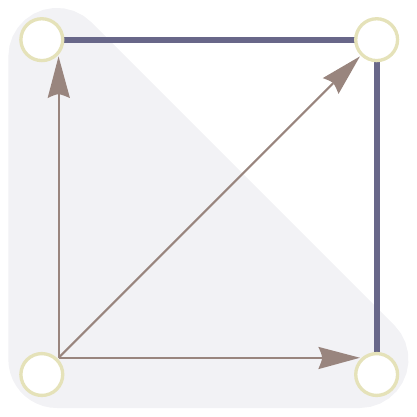}%
}\hspace{1.cm}
\subfloat[\label{subfig:c}]{%
  \includegraphics[width=0.42\columnwidth]{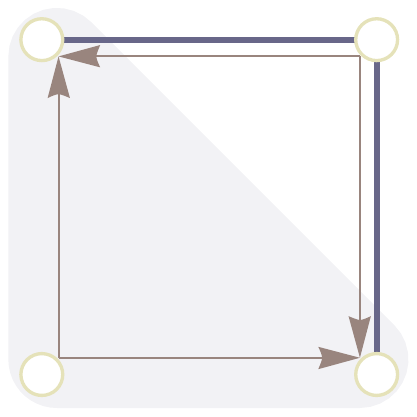}%
}
\caption{A network with $4$ vertices, two bipartite hyperedges (lines), one tripartite hyperedge (shaded rounded triangle) and two different LOCC protocols indicated by arrows. In protocol (a) one party implements a local quantum channel and sends classical information to others. In protocol (b) two parties implement local quantum channels and send classical information to others. While (a) consumes one round and two steps, (b) requires two rounds but only one step.}\label{fig:protocols}
\end{figure}

Then, for a given state $\rho$ and network $G$, the \textit{network round complexity} $\ecom(\rho | G)$ and the \textit{network communication cost} $\ecost(\rho | G)$  are defined to be the minimal average number of rounds and steps respectively of an LOCC protocol that prepares $\rho$ from a $G$ state with arbitrarily small error. That is,
\begin{align}
  {\cal E}_{r,c}(\rho|G) := \min_{p_i, \rho_i} \sum\nolimits_i p_i \epsilon_{r,c}(\rho_i|G),
\end{align}
where $\rho = \sum_i p_i \rho_i$ is any convex decomposition of $\rho$, and $\epsilon_{r,c}(\rho_i|G)$ is the minimal number of rounds or steps that are necessary to prepare $\rho_i$ to arbitrary accuracy from a $G$ state by LOCC. It follows from this definition that $\epsilon_{r,c}(\rho|G)=0$ (and, consequently, ${\cal E}_{r,c}(\rho|G)=0$) iff $\rho\in\scq(G)$. Notice that these measures take integer values for pure states, while this no longer needs to be the case for mixed states.

It turns out that the measures $\ecost$ and $\ecom$ are intimately connected to a well-known quantifier in quantum resource theories. We define the \textit{network-entanglement weight} as
\begin{align}
    \efr(\rho|G) = \min_{\sigma \in \scq(G)} p\ \text{ s.t. }\ \rho \succeq (1-p)\sigma,
\end{align}
which is an instance of the so-called best free approximation~\cite{regula2017convex} or weight of resource~\cite{ducuara2020weight}. All these measures are well-defined from the resource-theoretic point of view:
\begin{theorem}
    The network-entanglement measures $\ecost$, $\ecom$ and $\efr$ are convex, LOSR monotonic and subadditive.
\end{theorem}
The proof of this observation is provided in the App. II. These desirable properties of $\ecost$ and $\ecom$ do not imply, however, that these measures provide a fine grained quantification of network-entanglement. Do these measures depend effectively on the underlying network $G$? If so, how do they depend on the size and other properties of $G$? In fact, in the App. I, we show that $\efr$ can also be interpreted as the minimal complexity of preparing a state in a network by LOCC if we allow more than one parties to act non-trivially in a round. Yet, for all pure states $\rho$ it holds that $\efr(\rho|G)=1$ for every connected hypergraph $G$. Thus, in the following we consider a simple specific protocol to prepare any state in an arbitrary network{\color{black}, including any mixed states}. This automatically yields an upper bound for the communication cost and round complexity of network-entangled states in any given connected network by graph-theoretic parameters. We will later show that in many general instances these bounds are in fact optimal.

The main ingredient of the aforementioned protocol is teleportation. Since we do not impose constraints on the local capacity of each vertex and source capacity in the network, the parties connected by a hyperedge can share any quantum state among themselves. Thus, in particular, every pair of parties in a hyperedge can share an unbounded number of Bell pairs and, therefore, teleport to each other an unbounded number of qudits. Hence, we can prepare any target state $\rho$ locally on all vertices $v$ included in a chosen hyperedge $e$ and distribute the state by teleporting to neighbors in each step until we cover the whole network. Then, the particle for vertex $u$ is teleported sequentially via neighboring vertices from a vertex $v\in e$ until it reaches $u$. 
In this way, each vertex receives not only its own particle, but also the particles going through {it}, which are later distributed to other vertices. This protocol is exact and consists of $\max_u\min_{v\in e} \operatorname{dis}(u,v)$ steps, where $\operatorname{dis}(u,v)$ is the minimal number of hyperedges in a path connecting $u$ and $v$. By optimizing the choice of hyperedge $e$, we have that
\begin{align}
    \ecost(\rho|G)\! \le\! \epsilon_c(\rho|G)\! \le\! r_h(G), 
\end{align}
where $r_h(G) :=\! \min_{e\in E}\max_{u\in V}\min_{v\in e} \operatorname{dis}(u,v).$
Given its similarity to the radius, we term $r_h(G)$ the \textit{hyperedge radius} of the hypergraph $G$. Analogously, we obtain as well that $\ecom(\rho|G) \le \epsilon_r(\rho|G) \le d_c(G)$, where $d_c(G)$, which we refer to as the \textit{connected domination number}, stands for the minimal size of a connected subhypergraph \footnote{{Formally, given a hypergraph $G=(V,E)$, a subhypergraph of $G$ is any hypergraph $G'=(V',E')$ such that $V'\subseteq V$ and $E'\subseteq\{e\cap V':e\in E,|e\cap V'|\geq2\}$.}} such that any vertex is either in it or connected to a vertex in it. 
In App. III, we relate $r_h$ and $d_c$ to the notion of domination in graph theory~\cite{haynesDominationGraphsCore2023}. The averaged nature of $\ecom$ and $\ecost$ leads to the following result, {\color{black}whose proof is in App. III}.

\begin{theorem}\label{ob:bounds}
For any hypergraph $G$ and quantum state $\rho$, it holds that
    $\efr(\rho|G) \le \ecost(\rho|G)  \le r_h(G)\efr(\rho|G)$ and
    $\efr(\rho|G) \le \ecom(\rho|G)  \le d_c(G)\efr(\rho|G)$.
\end{theorem}

Remarkably, the above simple teleportation-based protocols cannot be improved in general, i.e.,\ the above upper bounds are tight for very general classes.
\begin{theorem}\label{ob:tight1}
    For any hypergraph $G$ in which there exists at least one pair of vertices that do not belong to the same hyperedge, and any pure state $\rho$ such that any bipartite reduced state for such a pair of vertices is entangled, then
      $\ecom(\rho|G) = d_c(G)$.
\end{theorem}
\begin{theorem}\label{ob:tight2}
  {Let $T$ be an arbitrary tree graph and let $d$ be its diameter, i.e., $d=\max_{u,v\in T} dis(u,v)$, and $\rho$ a pure state. If there exists at least one pair of parties at distance $d$ in $T$ for which the corresponding bipartite reduced state of $\rho$ is entangled, then  $\mathcal{E}_c(\rho|T)=r_h(T)$.}
\end{theorem}
{\color{black}The proofs of those two observations are provided in App. IV and V, respectively}.
This additionally shows that the measures are not trivial in the sense that a limited number of steps or rounds is sufficient in general to prepare any state. It follows from the above observations that both measures can grow unboundedly with the size of the network. 
This raises the question of whether the values of $\ecost$ and $\ecom$ are given by $r_h(G)$ and $d_c(G)$ for \textit{all} $G$-entangled pure states. Interestingly, the answer is negative. There exist pure states and networks which do not saturate these upper bounds and, in fact, the lower bounds given in Observation 2 can be achieved. Thus, the network does not fix the value of the measures for all $G$-entangled pure states. The details are given in App. IV.

{
\color{black}
In the above protocol, all the measurements in the procedure of teleportation commute with each other, since they act on different particles. Thus, they can be implemented simultaneously. If there is no restriction on classical communication of the outcomes, then the whole protocol can be finished in one run, i.e., all the nodes only communicate once. This leads to another measure of network-entanglement named as \textit{network fraction}, that is,
\begin{equation}
  {\cal E}_f(\rho|G) := \min_{p_i,\rho_i} \sum_i p_i r_f(\rho_i|G),
\end{equation}
 where $r_f(\rho_i)$ is the number of runs of LOCC operations to create $\rho_i$ deterministically, without any limitation on the  classical message distribution.
As it turns out,
\begin{align}
    {\cal E}_f(\rho|G)=\efr(\rho|G).
\end{align}
The proof and more details are in App. I.
}

\section{Estimation of network-entanglement measures}
From our analysis above, $\efr(\rho|G)$ plays an important role in the estimation of $\ecost(\rho|G)$ and $\ecom(\rho|G)$.
Naturally, $\efr(\rho|G) = 1$ for any $G$-entangled pure state $\rho$. This property holds as well for special classes of mixed states, {\color{black} as proven in App. VI}.

\begin{theorem}
For any anti-symmetric state $\rho$, $\efr(\rho|G) = 1$. For any symmetric state $\rho$, $\efr(\rho|G)$ does not depend on the network topology $G$, {\color{black}i.e. the network-entanglement weight of $\rho$ is the entanglement weight of $\rho$.}
\end{theorem}
The exact computation of our measures, especially for general mixed states, is a hard problem, which, in particular, requires a complete characterization of $\scq(G)$. In the following, for the case of arbitrary states, we provide techniques for the estimation of $\efr$ based on witnesses and covariance inequalities. In turn, these results can be applied to $\ecost$ and $\ecom$ via Observation~\ref{ob:bounds}. To ease the notation, in the following we omit the dependence on $G$ of the measures wherever it is clear from the context.
\subsection{Estimation based on witnesses} 
The witness method has been used extensively in the detection of quantum entanglement~\cite{terhal2000bell}, and also in the estimation of entanglement measures~\cite{brandao2005quantifying,cavalcanti2006estimating,huber2013structure}. 
We show in the following that this is also the case for network-entanglement, where the essential step is to construct efficient witnesses.
In the App. IX, we show that  
for any network with at most $k$-partite sources (i.e.,\ a $k$-network), it holds that 
\begin{equation}
\max_{\sigma \in \scq} F(\operatorname{GHZ}, \sigma) \le k_d/(k_d+1),
\end{equation}
where $\operatorname{GHZ} = \sum_{i,j=0}^{d-1} (|i\rangle\langle j|)^{\otimes n}/d$, $k_d = (1+\sqrt{dk})^2/(d-1)$ with $d$ and $n$ to be the local dimension and the size of network, respectively. This leads to the $k$-network-entanglement witness
\begin{align}
    W_{k,d} = {k_d} \id - (k_d+1)\operatorname{GHZ}.
\end{align}
Denoting $w_{k,d}(\rho) := \max \{0, - \tr(W_{k,d}\,\rho)\}$ and observing that $W_{k,d}\geq-\id$, it follows from well-known properties of this quantifier in general resource theories~\cite{regula2017convex} that for all $k$-networks
\begin{equation}
    \efr(\rho) \ge w_{k,d}(\rho).
\end{equation}
This estimation is tight for the GHZ state, since $w_{k,d}(\operatorname{GHZ}) = 1$. 

\subsection{Estimation based on covariance}
For a given state $\rho$, let us denote by $\Gamma(\rho)$ the covariance matrix relative to some choice of dichotomic measurements $\{M_i\}$, where $i$ refers to the $i$-th party, i.e., $\Gamma_{ij}(\rho) = \tr(M_iM_j\rho) - \tr(M_i\rho)\tr(M_j\rho)$.
Moreover, let $\omega_k(\rho) = \sum_{ij} \Gamma_{ij}(\rho) - k \tr(\Gamma(\rho))$,
and $\beta(\rho) = 2\sqrt{1-[\tr(\rho^2)]^2}$. Then building on the network-entanglement criterion of Ref.~\cite{xu2023characterizing}, we show in App. VII that for $k$-network-entangled states it holds that
\begin{align}\label{eq:cov_lower}
    \efr(\rho) \ge \frac{\omega_k(\rho)}{n(n-k)} - \beta(\rho),
\end{align}
The estimation in Eq.~\eqref{eq:cov_lower} is also tight for the GHZ state considering only the Pauli measurement $Z$, i.e., $M_i=Z$ for all $i$.

\begin{figure}[th]
  \includegraphics[width=1.0\columnwidth]{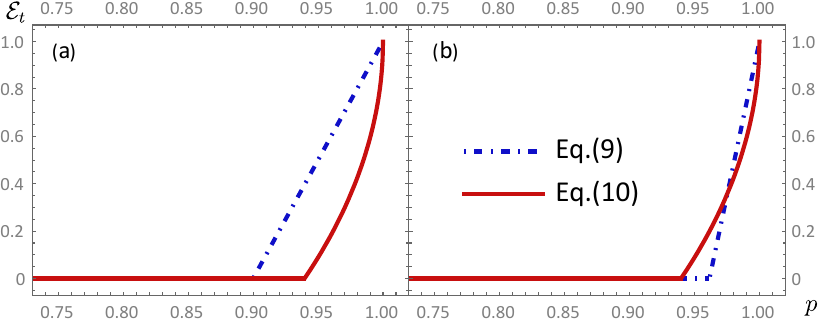}%
\caption{Estimation of $\efr(p{\rm GHZ} + (1-p)\id/d)$ for $10$-partite $k$-network-entanglement with the bounds given in Eqs.\ (9) and (10), where $k=2$ in (a) and $k=8$ in (b).}\label{fig:comparison}

\end{figure}

{For the sake of illustration}, we plot in Fig.~\ref{fig:comparison} the lower bounds on $\efr(\rho)$ for $\rho = p\operatorname{GHZ} + (1-p)\id/2^{10}$ to be a $10$-partite state with local dimension $2$ for different values of $k$.

In addition to this, a see-saw method based on semi-definite programming is presented in App. VIII, which provides estimations from above and works well for small networks.

\section{Conclusions}
We have addressed the quantification of network-entanglement and we have provided different such measures. In particular, the network-entanglement weight, the network communication cost and the network round complexity have a clear operational interpretation in this scenario as the cost of preparing network-entangled states by LOCC. This, in turn, can be related to the quality of local quantum memories and to the information flow in the network. This is rather different from the case of standard entanglement theory, where LOCC operations are considered {to be} free, and stems from the fact that network-entanglement is a resource theory under LOSR. We have also shown that the latter two measures are closely connected to certain graph-theoretic quantities and to the network-entanglement weight, and we have used this to devise methods to estimate them by developing more efficient witnesses.

In this work, we have assumed no restriction on the classical or quantum capacities of the network links, an issue that might be relevant in certain practical scenarios. We leave for future study the construction of network-entanglement measures that take this into account and their relation to the ones proposed here. We hope that our work paves the way to a fully-fledged resource theory of network entanglement that could eventually lead to new features and applications of entanglement in the network setting.  

\begin{acknowledgments}
Z.-P. X. acknowledges support from {National Natural Science Foundation of China} (No.\ 12305007), Anhui Provincial Natural Science Foundation (No.\ 2308085QA29), Anhui Province Science and Technology Innovation Project (No. 202423r06050004). J. I. de V. acknowledges financial support from the Spanish Ministerio de Ciencia e Innovaci\'on (grant PID2023-146758NB-I00  and ``Severo Ochoa Programme for Centres of Excellence'' grant CEX2023-001347-S funded by MCIN/AEI/10.13039/501100011033) and from Comunidad de Madrid (grant TEC-2024/COM-84-QUITEMAD-CM and the Multiannual Agreement with UC3M in the line of Excellence of University Professors EPUC3M23 in the context of the V PRICIT). L.L. S. and S. Y. would like to thank Key-Area Research and Development Program of Guangdong Province (No. 2020B0303010001) and Innovation Program for Quantum Science and Technology (2021ZD0300804). 
\end{acknowledgments}

\onecolumngrid
\begin{center}
  \huge \sf{Appendix}
\end{center}
\vspace{1em}
\renewcommand*{\thesection}{\Roman{section}}
\setcounter{section}{0}
Throughout this document, {we follow the same notation established in the main text. Given a hypergraph $G$, $V$ will always denote the set of its vertices and $E$ the set of its hyperedges (whose elements are subsets of $V$ of cardinality at least 2.). If all the elements in $E$ have cardinality equal to two, then we refer to $G$ as a graph and to the elements of $E$ as edges. The distance $dis(i,j)$ of two vertices $i,j\in V$ is the minimal number of hyperedges connecting them. By $N(v)$, we denote the set of vertices that are adjacent to vertex $v\in V$ in a given hypergraph. The radius of a hypergraph is given by $\min_{v\in V} \max_{u\in V} dis(u,v)$ and the diameter by  $\max_{v\in V} \max_{u\in V} dis(u,v)$.}

\section{Network-entanglement weight}
{In this section, we provide an interpretation of the network-entanglement weight as some form of complexity of preparing the state in the network by LOCC operations. Measuring the complexity of an LOCC protocol by \textit{rounds} puts the main constraint in the fact that only one party is allowed to send classical information per round, while the main constraint when we use \textit{steps} is that classical information can only flow to neighbors in each step. Thus, arguably, the loosest way in which we can measure the complexity of an LOCC protocol is by lifting these two limitations, which leads to what we refer in the following as \textit{runs}. Thus, in one run an arbitrary number of parties can implement local quantum measurements and the corresponding outcomes can be transmitted to any other parties. The only limitation arises from the fact that the particular local POVMs to be implemented by some parties may depend on the classical information produced by previous local measurements, which gives rise to the notion of run. More precisely, when an LOCC protocol is decomposed as a sequence of runs, in each run any number of parties can implement a local quantum channel conditioned on the classical information produced at previous runs. These parties can then send to all other parties classical information based on the outcome of their channels. Local correction channels (wich cannot be postselected and produce no classical information) can then be applied by any party and after this a new run can start. See Fig.~\ref{fig:fc} to see the illustration of the concepts \textit{rounds}, \textit{steps} and \textit{runs}. The main observation in this section is as following.}

\begin{figure}[htpb]
\centering
\subfloat[\label{subfig:b}]{%
  \includegraphics[width=0.22\columnwidth]{12a2.pdf}%
}\hspace{1.5cm}
\subfloat[\label{subfig:c}]{%
  \includegraphics[width=0.22\columnwidth]{ns2.pdf}%
}\hspace{1.5cm}
\subfloat[\label{subfig:a}]{%
  \includegraphics[width=0.22\columnwidth]{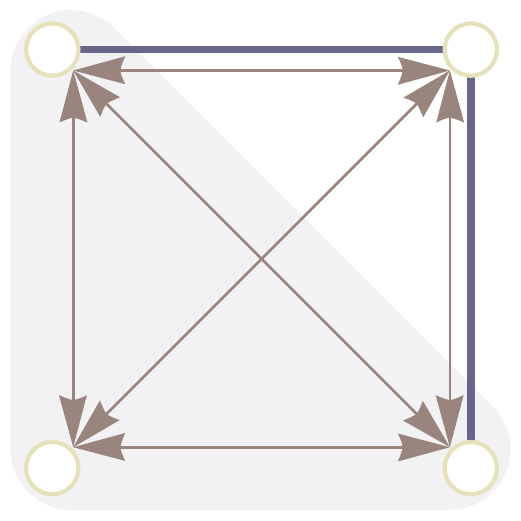}%
}
\caption{{\color{black}A network with 4 vertices, two bipartite hyperedges (lines) and one tripartite hyperedge (shadowed rounded triangle), where the arrows indicate the communication direction of classical information after parties implement local operations. In protocol (a) one party implements a local quantum channel and sends classical information to all others, who implement local unitary corrections. This consumes one round, two steps and one run. In protocol (b) two parties implement local quantum channels and send classical information to their neighbors, who implement local unitary corrections. This consumes two rounds, one step and one run. In protocol (c) all parties implement local quantum channels and broadcast classical information, which is then used by all parties to implement local unitary corrections. This consumes four rounds, two steps and one run.}}\label{fig:fc}
\end{figure}

\setcounter{theorem}{-1}
\begin{theorem}\label{ob:0}
 {In a network described by a connected hypergraph, any network-entangled state can be generated with an LOCC protocol consisting of a single run.}
\end{theorem}

We delay the proof of this observation to the end of this section. 

{Let $r_f(\rho|G)$ denote the number of runs to prepare the state $\rho$ to arbitrary precision in the network $G$ with an optimal LOCC protocol. The above result means that $r_f(\rho|G)=1$ if $\rho$ is $G$-entangled and $r_f(\rho|G)=0$ otherwise. This implies that the network entanglement weight, which we have defined to be
\begin{equation}
	\efr(\rho|G) := \min_{\sigma\in \scq(G)} p\ s.t.\ \rho \succeq (1-p)\sigma, 
\end{equation}
is equivalent to 
\begin{equation}
  {\cal E}_f(\rho|G) :=\min_{w_i,\rho_i} \sum_i w_i r_f(\rho_i|G),
\end{equation}
where the minimization runs over all states $\{\rho_i\}$ and convex weights $\{w_i\}$ such that $\rho=\sum_iw_i\rho_i$. Thus, $\efr(\rho|G)$ can be interpreted as the average number of runs to prepare the state $\rho$ in the network $G$.

To see this,  notice first that for any $\sigma \in \scq(G)$ and $p\in [0,1]$ such that $\rho \succeq (1-p)\sigma$, we have 
\begin{equation}
  \rho = p \sigma' + (1-p)\sigma,
\end{equation}
where $\sigma' = (\rho - (1-p)\sigma)/ p \succeq 0$ by assumption.
Since $r_f(\sigma') \le 1$ and $r_f(\sigma)=0$, we have ${\cal E}_f(\rho|G) \le p$ by definition and consequently ${\cal E}_f(\rho|G) \le \efr(\rho|G)$. Secondly, for any convex decomposition $\rho = \sum_{i} w_i \rho_i$,
denote $p(\{w_i,\rho_i\})=\sum_{i,\rho_i\not\in \scq(G)} w_i$ and $\sigma(\{w_i,\rho_i\})=\sum_{i,\rho_i\in \scq(G)} w_i\rho_i/[1- p(\{w_i,\rho_i\})]$.
Clearly, $\rho \succeq [1-p(\{w_i,\rho_i\})]\,\sigma(\{w_i,\rho_i\})$. Hence, $p(\{w_i,\rho_i\}) \ge \efr(\rho|G)$.
By definition, 
\begin{equation}
	{\cal E}_f(\rho|G) = \min_{w_i,\rho_i} \sum_i w_i r_f(\rho_i|G) = \min_{w_i,\rho_i} p(\{w_i,\rho_i\}) \ge \efr(\rho|G).
\end{equation}

Before we prove Observation~\ref{ob:0}, we introduce our basic techniques in the proof.

In the $d$-dimensional case, denote $\Psi = |\psi\rangle\langle \psi|$ and $|\psi\rangle = \sum_{i=1}^d |ii\rangle/ \sqrt{d}$.
Denote $B_{ti_t} = U^{(i_t)}_{s_t}\Psi_{r_ts_t} {U^{(i_t)}_{s_t}}^\dagger$, where $U^{(i)}_k = U^{(i)}$ is a unitary acting on the particle with label $k$. We remark that $\{U^{(i)}\}_{i=1}^{d^2}$ are generalized Pauli matrices in $d$-dimensional space, which form the group ${\cal G}_d$ with respect to matrix product~\cite{patera1988pauli}. 
To be more explicit, all the elements in ${\cal G}_d$ are $\{\sigma_{k,j}\}_{k,j=1}^d$, where
\begin{equation}
\sigma _{k,j}=\sum _{m=0}^{d-1}|m+k\rangle \omega ^{jm}\langle m|,\ \omega =\exp\Big(\frac{2\pi i}{d}\Big).
\end{equation}
Notice that any two different elements in the group ${\cal G}_d$ are orthogonal in the sense of the Hilbert-Schmidt inner product, i.e.\ $\tr[(\sigma _{k,j})^\dagger\sigma _{k',j'}]=0$ if $k\neq k'$ or $j\neq j'$.
Then, by definition, $B_{ti_t}$'s are orthogonal projectors for different $i_t$.

{\color{black}
Hereafter, {in a slight abuse of notation we do not distinguish} $M_k$ from $M_k \otimes \id_{\bar{k}}$ where $\bar{k}$ are the labels for other parties than $k$. We firstly mention some basic mathematical results:
\begin{itemize}
    \item $U_1 |\psi\rangle = U_2^T|\psi\rangle$, which implies $U_1 \Psi_{12} U_1^\dagger = U_2^T\Psi_{12}U_2^*$, where $U_1$ and $U_2$ are the same unitary $U$ but act on parties $1$ and $2$ respectively.
    \item $\tr_{23}[M_{12} \Psi_{34}  \Psi_{23} ] = M_{14}$ and $\tr_{23}[ (M_{12} \Psi_{34}) (U_3 \Psi_{23} U_3^\dagger) ] = U_4^* M_{14} U_4^T$, where $M_{12}$ and $M_{14}$ are the same matrix but supported by different parties.
\end{itemize}
For the completeness, we provide the proof here. Denote $u_{ij}$ to be the $(i,j)$-th element of the unitary $U$. Direct calculation shows that
\begin{align}
    &U_1 |\psi\rangle \propto \sum_{i,j,k} u_{ij} \left(|i\rangle\langle j|\otimes \id \right) |kk\rangle = \sum_{i,j,k} u_{ij}  |ik\rangle \delta_{jk} = \sum_{i,j} u_{ij}  |ij\rangle,\\
     &U_2^T |\psi\rangle \propto \sum_{i,j,k} u_{ij} \left(\id \otimes |j\rangle\langle i| \right) |kk\rangle = \sum_{i,j,k} u_{ij}  |kj\rangle \delta_{ik} = \sum_{i,j} u_{ij}  |ij\rangle.
\end{align}
Assume the matrix $M = \sum_{ijkl} m_{ijkl} |i\rangle\langle j| \otimes |k\rangle\langle l|$. Then
\begin{align}
    \tr_{23}[ M_{12} \Psi_{34}  \Psi_{23} ] = &\sum_{ijklpqtr} m_{ijkl} \langle rr|_{23}|ik\rangle\langle jl|_{12} |pp\rangle\langle qq|_{34} |tt\rangle_{23}\\
    = &\sum_{ijklpqtr} m_{ijkl} \delta_{rk}\delta_{lt}\delta_{rp}\delta_{qt} |i\rangle\langle j|_{1} |p\rangle\langle q|_{4} \\
    = &\sum_{ijkl} m_{ijkl} |i\rangle\langle j|_{1} |k\rangle\langle l|_{4} = M_{14}.
\end{align}
Consequently, we have
\begin{align}
    &\tr_{23}[ (M_{12} \Psi_{34}) (U_3 \Psi_{23} U_3^\dagger) ] = \tr_{23}[ M_{12} (U_3^\dagger\Psi_{34} U_3) \Psi_{23}] = \tr_{23}[ M_{12} (U_4^*\Psi_{34} U_4^T) \Psi_{23}]\nonumber\\
    = &U_4^*\tr_{23}[ M_{12} \Psi_{34} \Psi_{23}] U_4^T = U_4^* M_{14} U_4^T.
\end{align}
If we want to teleport a particle from one end of a line to another end, c.f.\ the example in Fig.~\ref{fig:telepath}, a typical way is to do the teleportation between neighbor nodes sequentially. We claim that this can also be done in one run by all nodes sending classical communication to the last one.
\begin{figure}[thb]
  \centering
  \includegraphics[width=0.55\textwidth]{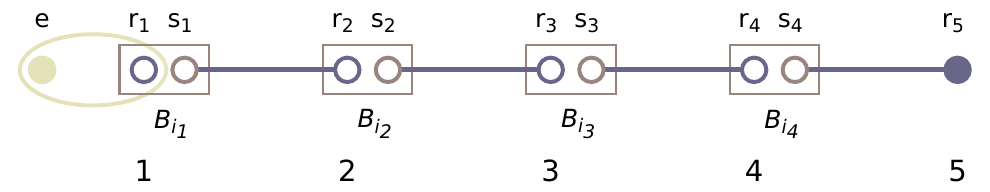}
  \caption{A path with $5$ nodes, where $B_{i_k}$ is the projection corresponding outcome $i_k$ of the Bell measurement, the source containing particles $s_kr_{k+1}$ is the maximally entangled state for each $k$. The particle $r_1$ is initially entangled with an extra system labeled by $e$.}\label{fig:telepath}
\end{figure}
For the following derivation, one can examine it with the example in Fig.~\ref{fig:telepath}, where the task is to teleport the state of $r_1$, which is in correlation with an extra party $e$, to $r_5$. 
Denote $\sigma_{i_1i_2\ldots i_n}$ the post-measurement state in the registers $e$ and $r_{n+1}$ after obtaining measurement outcomes $i_1, i_2, \ldots, i_n$ from the sequential Bell measurements on the initial state $\rho_{er_1}$.
Then we have,
\begin{align}
  \sigma_{i_1i_2\ldots i_n}
  \propto &\tr_{1,\ldots,n} \Big[\rho_{er_1}\Big(\prod\nolimits_{k=1}^n \Psi_{s_kr_{k+1}}\Big) \Big( \prod\nolimits_{k=1}^n B_{k i_k}\Big)\Big]\\
  =&\tr_{2,\ldots,n} \Big[\tr_1\Big(\rho_{er_1}\Psi_{s_1r_{2}}U^{(i_1)}_{s_1}\Psi_{r_1 s_1} {U^{(i_1)}_{s_1}}^\dagger \Big)\Big(\prod\nolimits_{k=2}^n \Psi_{s_kr_{k+1}} \Big) \Big(\prod\nolimits_{k=2}^n U^{(i_k)}_{s_k}\Psi_{r_k s_k} {U^{(i_k)}_{s_k}}^\dagger\Big)\Big]\\
  \propto&\tr_{2,\ldots,n} \Big[\Big({U^{(i_1)}_{r_2}}^*\rho_{er_2}{U^{(i_1)}_{r_2}}^T\Big)\Big(\prod\nolimits_{k=2}^n \Psi_{s_kr_{k+1}} \Big) \Big(\prod\nolimits_{k=2}^n U^{(i_k)}_{s_k}\Psi_{r_k s_k} {U^{(i_k)}_{s_k}}^\dagger\Big)\Big]\\
  &\vdots\\
  \propto&\tr_{n} \Big[{\Big(\prod\nolimits_{k=n-1}^1 U^{(i_k)}_{r_n}}\Big)^*\rho_{er_n}\Big(\prod\nolimits_{k=n-1}^1 {U^{(i_1)}_{r_n}}\Big)^T \Big( \Psi_{s_n r_{n+1}} \Big) \Big( U^{(i_n)}_{s_n}\Psi_{r_n s_n} {U^{(i_n)}_{s_n}}^\dagger\Big)\Big]\\
  \propto& \Big(\prod\nolimits_{k=n}^1 U^{(i_k)}_{r_{n+1}}\Big)^* \rho_{er_{n+1}} \Big(\prod\nolimits_{k=n}^1 {U^{(i_1)}_{r_{n+1}}}\Big)^T,
  \label{eq:tobecontinued}
\end{align}
where the first line is by definition of $\sigma_{i_1i_2\ldots i_n}$, the second line is from the definition of $B_{ki_k}$, the third to the last lines hold due to the aforementioned mathematical results.

Consequently, we know that
\begin{equation}
  \rho_{er_{n+1}} = V^T_{r_{n+1}} 
  \sigma_{i_1i_2\ldots i_n}
  V^*_{r_{n+1}},
\end{equation}
where $V=U^{(i_n)} U^{(i_{n-1})}\cdots U^{(i_1)}$ is a unitary totally determined by $i_1, \cdots, i_n$. Since all $U^{(i_t)}$ are in the group ${\cal G}_d$, $V$ is also one element in this group. This entails that we only need to implement one correction with a unitary in the group ${\cal G}_d$ instead of a sequence of unitaries, thus the whole procedure consumes only one run.
}

Then we proceed to prove Observation~\ref{ob:0}.
\begin{proof}
  {Without loss of generality, we only need to prove the statement for the case in which $G$ is a tree graph. This is because any source in the network can produce arbitrary bipartite states between all pairs of parties connected by this source. Thus, any hyperedge can be transformed into edges connecting all pairs of vertices in the hyperedge. In this way, from any given connected hypergraph, we can produce a connected graph, which will always contain a tree, from which the above claim follows.}
  
  Then we can choose the node denoted by $1$ as the root of the tree and generate the state $\rho$ locally on this node firstly. Notice that there is always a path $P_k$  to connect $1$ and any other node $k$ in $G$ since $G$ is a tree by assumption. Since there is no limitation on the power of sources and local nodes, we can do the teleportation protocol along paths $\{P_k\}_{k\in V(G)}$ independently at the same time {(See Fig.~\ref{fig:tree} for an illustration).}

\end{proof}

\begin{figure}[thb]
	\centering
\subfloat[\label{tree2}]{%
  \includegraphics[width=0.22\textwidth]{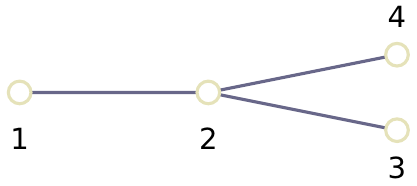}
}\hspace{4em}
\subfloat[\label{loccnet}]{%
  \includegraphics[width=0.4\columnwidth]{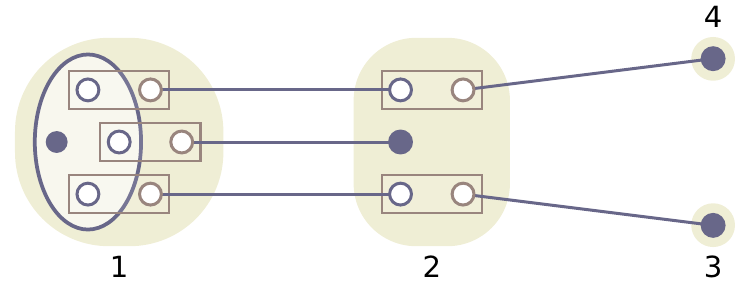}%
}
\caption{(a) A tree network with $4$ nodes; (b) The paths to connect node $1$ to any other nodes, along which the entangled particles (in the big blue circle) prepared in node $1$  are teleported to other nodes.}\label{fig:tree}
\end{figure}

\section{Proof of Observation 1}
\textbf{Observation 1.}\textit{
    The network-entanglement measures $\ecost$, $\ecom$ and $\efr$ are convex, LOSR monotonic and subadditive.}
\begin{proof}
    We prove the desired properties for $\ecost$. The cases of $\ecom$ and $\efr$ follow analogously.

    \textit{Convexity:} For a given set of states $\{\rho_t\}_t$, denote by $\rho = \sum_t p_t \rho_t$ an arbitrary convex combination of them. Assume that  
    $\ecost(\rho_t) =  \sum\nolimits_{k_t} p_{k_t} \epsilon_c(\rho_{k_t})$,
    with $\rho_t = \sum_{k_t} p_{k_t} \rho_{k_t}$. This induces a convex decomposition of $\rho$, that is, $\rho = \sum_t \sum_{k_t} p_t p_{k_t} \rho_{k_t}$. By definition,
    $ \ecost(\rho) \le \sum\nolimits_{t,k_t} p_t p_{k_t} \epsilon_c(\rho_{k_t}) = \sum\nolimits_{t} p_t \ecost(\rho_t)$.
The convexity of $\efr$ and $\ecom$ can be proven similarly.    
    
    \textit{Subadditivity:} {\color{black}Assume that  
    $\ecost(\rho_t) =  \sum\nolimits_{k_t} p_{k_t} \epsilon_c(\rho_{k_t})$,
    with $\rho_t = \sum_{k_t} p_{k_t} \rho_{k_t}$. Let $\rho = \otimes_{t=1}^n \rho_t$.} The subadditivity of $\ecost$, that is $\ecost(\rho) \le \sum_t \ecost(\rho_t)$, follows from the fact that we can act separately and simultaneously on the $\rho_t$'s. Indeed, since $\rho = \otimes_{t=1}^n \sum_{k_t} p_{k_t}\rho_{k_t}$ and $\ecost(\rho_t)=\sum_{k_t}p_{k_t}\epsilon_c(\rho_{k_t})$, then
    \begin{align}
    \ecost(\rho)&\leq \sum_{k_1, \ldots, k_n} \left[\prod_{l=1}^n p_{k_l}\right]  \epsilon_c\left(\bigotimes_{t=1}^n\rho_{k_t}\right)\leq
      \sum_{k_1, \ldots, k_n} \left[\prod_{l=1}^n p_{k_l}\right] \max_t \epsilon_c(\rho_{k_t})\leq\sum_{k_1, \ldots, k_n} \left[\prod_{l=1}^n p_{k_l}\right] \sum_t \epsilon_c(\rho_{k_t})\nonumber\\
      &=\sum_{t=1}^n\sum_{k_1, \ldots, k_n} \left[\prod_{l=1}^n p_{k_l}\right] \epsilon_c(\rho_{k_t})=\sum_{t=1}^n \sum_{k_t} p_{k_t} \epsilon_c(\rho_{k_t}) = \sum_{t=1}^n \ecost(\rho_t).
    \end{align}
The same argument holds also for $\efr$ and $\ecom$ if we replace $\epsilon_c$ by $\epsilon_w$ and $\epsilon_r$, respectively.

\textit{LOSR monotonicity:} Since $\ecost(\rho)$ is convex, we only need to show that $\ecost(\rho)$ does not increase under local operations.
    Denote ${\cal C}$ the channel corresponding to local operations and let $\rho = \sum_i p_i \rho_i$ be the decomposition such that $\ecost(\rho) = \sum_i p_i \epsilon_c(\rho_i)$. Notice that $\epsilon_c({\cal C}(\rho_i)) \le \epsilon_c(\rho_i)$ $\forall i$ since we can always attach the local operations to the last round of the LOCC protocol. This implies that
    $\ecost(\rho) \ge \sum_i p_i \epsilon_c({\cal C}(\rho_i)) \ge \ecost({\cal C}(\rho))$.
The proof for $\efr(\rho)$ and ${\cal E}_r(\rho)$ follows along the same lines.
\end{proof}

\section{Graph parameters, teleportation protocols and proof of Observation 2}

A given subset $S$ of vertices of a hypergraph $G$ is \textit{dominating}~\cite{haynesDominationGraphsCore2023} if every vertex not in $S$ has a neighbor in $S$. A subgraph $G'$ is \textit{dominating} if the corresponding vertex set is, denoted as $G' \subset_d G$. The connected domination number $d_c$ from the main text boils down then to the minimal size of a connected dominating subgraph. The hyperedge radius is related as well to the notion of domination. Namely, it holds that
\begin{align}\label{eq:dominations}
    r_h(G) = & \min k,\nonumber\\
    s.t.\ & G_i \subset_d G_{i+1}, \forall i=0, \ldots, k-1\nonumber\\
          & G_k = G,\ G_{0} \text{ is a hyperedge}.
\end{align}

To see this, notice first that if $r_h(G) = k'$, then we can denote by $e$ one hyperedge achieving the value $k'$, and by $G_i$ the subgraph of $G$ induced by all the vertices $u$ such that $\min_{v\in e} dis(u,v) \le i$. By definition of the distance, we know that $G_{i} \subset_d G_{i+1}$ and $G_0 = e$ is a hyperedge. By definition of the hyperedge radius, we have $G_{k'} = G$. Thus, we have one feasible instance in Eq.~\eqref{eq:dominations}. Conversely, assume now that $\{G_i\}_{i=1}^k$ is a sequence of subgraphs such that $G_i \subset_d G_{i+1}$, $G_k = G$ and $G_{0}=e$ is a hyperedge. Since any vertex in $G_{i+1}$ is either included in $G_i$ or connected to at least one vertex in $G_i$, by definition of domination, any vertex in $G$ can be connected to the hyperedge $e$ with at most $k$ hyperedges. That is, $r_h(G) \le k$. This proves Eq.\ (\ref{eq:dominations}).

{To help the reader understand the teleportation-based protocol and that its number of steps and rounds is respectively $r_h(G)$ and $d_c(G)$, we present an example in a network corresponding to the hypergraph with $7$ vertices shown in Fig.~\ref{fig:tel12} with the aim of preparing a 7-qubit state.} The whole protocol goes as follows:
\begin{enumerate}
    \item[(a)] {The bipartite source for parties $1,2$ prepares a Bell pair and the bipartite source for parties $4,5$ prepares three Bell pairs. The tripartite source for parties $5,6,7$ prepares two Bell pairs for parties $5,6$ and parties $5,7$, respectively.
	\item[(b)] The tripartite source for parties $2,3,4$ prepares the target state for all the $7$ parties, where the particles for parties $1, 2$ are firstly collected by party $2$, the one for party $3$ by itself, and the ones for parties $4,5,6,7$ by party $4$.
	\item[(c)] The particle for party $1$ is then teleported from party $2$ to party $1$ consuming the corresponding Bell pair.
	\item[(d)] The particles for parties $5,6,7$ are teleported from party $4$ to party $5$ consuming the corresponding three Bell pairs.
	\item[(e)] Party $5$ teleports the particles for party $6$ and $7$ respectively, consuming in each case the corresponding Bell pair.}
\end{enumerate}
{This LOCC protocol consists of three rounds (corresponding to (c), (d) and (e) above) as only parties $2,4,5$ act non-unitarily and must send classical information to the rest. On the other hand, it only consumes two steps as all classical communication needs to be sent only to neighbors and part (c) and (d) of the protocol can be grouped in a single step. Notice that, indeed, the hyperedge radius $r_h$ of this hypergraph is $2$ and the connected domination number $d_c$ is $3$, as one can readily verify.}
We move now to the proof of Observation 2.

\begin{figure}[th]
  \centering
\subfloat[\label{subfig:tel1}]{%
  \includegraphics[width=0.4\columnwidth]{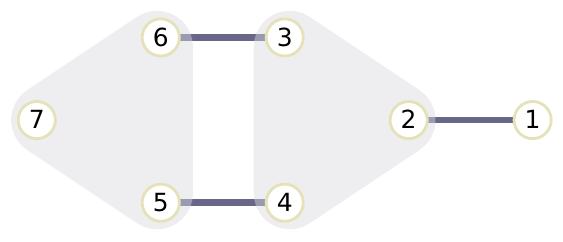}%
}\hspace{5em}
\subfloat[\label{subfig:tel2}]{%
  \includegraphics[width=0.4\columnwidth]{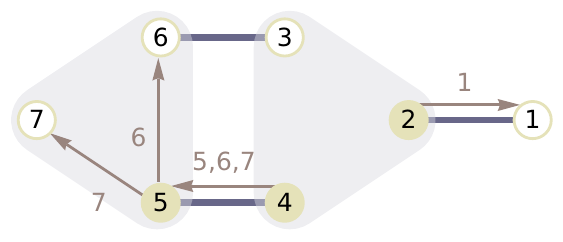}%
}
\caption{A network with $2$ tripartite sources and $3$ bipartite sources in (a) and the teleportation protocol in (b), where the arrows indicate the teleportation direction and the labels of the arrows indicate the teleported particles for the corresponding parties. The initial state for all parties is prepared by the tripartite source including vertices $2, 3, 4$. Only vertices $2, 4, 5$ { implement non-unitary local quantum channels throughout the nested teleportation scheme.}
}\label{fig:tel12}
\end{figure}

\noindent\textbf{Observation 2.}\textit{
For any hypergraph $G$ and quantum state $\rho$, it holds that
    $\efr(\rho|G) \le \ecost(\rho|G)  \le r_h(G)\efr(\rho|G)$ and
    $\efr(\rho|G) \le \ecom(\rho|G)  \le d_c(G)\efr(\rho|G)$.
    }
\begin{proof}
	We give the exact proof for ${\cal E}_c(\rho|G)$ and the case of ${\cal E}_r(\rho|G)$ follows exactly the same reasoning.

  In the main text we have shown that $1\le \epsilon_c(\rho|G) \le r_h(G)$ for every $G$-entangled state $\rho$. Thus, $\epsilon_w(\rho|G)\le \epsilon_c(\rho|G) \le r_h(G)\epsilon_w(\rho|G)$ holds for any state $\rho$. Now, by definition of ${\cal E}_c(\rho|G)$, there is a decomposition $\rho = \sum_i p_i \rho_i$ such that ${\cal E}_c(\rho|G) = \sum_i p_i \epsilon_c(\rho_i|G) \ge \sum_i p_i \epsilon_w(\rho_i|G)$. Thus, $\efr(\rho|G) \le {\cal E}_c(\rho|G)$.   
  On the other hand, by definition of $\efr(\rho|G)$, there is a decomposition $\rho = \sum_i p_i \rho_i$ such that $r_h(G)\efr(\rho|G) = \sum_i p_i \epsilon_w(\rho_i|G) r_h(G) \ge \sum_i p_i \epsilon_c(\rho_i|G)\ge {\cal E}_c(\rho|G)$. 
\end{proof}

\section{Proof of Observation 3}
\noindent\textbf{Observation 3.}\textit{
   For any hypergraph $G$ in which there exists at least one pair of vertices that do not belong to the same hyperedge, and any pure state $\rho$ such that any bipartite reduced state for such a pair of vertices is entangled, then
      $\ecom(\rho|G) = d_c(G)$.
    }

\begin{proof}

  Notice that $\efr(\rho|G) \le \ecom(\rho|G) \le d_c(G) \efr(\rho|G)$.
  In the case that $d_c(G) = 1$, we only need to prove that $\efr(\rho|G) = 1$, that is, the pure state $\rho$ is $G$-entangled. This is indeed the case since any bipartite reduced state $\tau_{ij}$ for non-adjacent vertices $i$ and $j$ (which must exist by assumption) of a $G$ state $\tau$ is separable. Therefore, so must be every state obtained by LOSR from $\tau$ and, since the set of separable states is closed, the same applies to any state in $\scq(G)$, while $\rho_{ij}$ is entangled. 

  We turn now to consider the case that $d_c(G) \ge 2$. 
  Then for any vertex $i$, there is always another vertex $p_i$ such that those two vertices share no common source (i.e.\ they are not adjacent). Otherwise, the vertex $i$ is connected to any other {vertex and we would have that $d_c(G)=1$.}
  This implies that the reduced state $\tau_{i\, p_i}$ is separable for any $G$ state $\tau$. 
  Using the same argument as before, a state $\rho$ fulfilling the hypothesis of this observation must then be $G$-entangled. We claim that such a pure state $\rho$ cannot be prepared to arbitrary accuracy from $G$ states with round complexity strictly less than $d_c(G)$, i.e.\ $\epsilon_r(\rho|G)>d_c(G)-1$.
  For a LOCC operation $\Lambda$ with round complexity $d_c(G)-1$, let us denote by $S$ the set of vertices which send information (hence, $|S|=d_c(G)-1$) and by ${\cal N}(S)$ the set of vertices either in $S$ or being neighbors to at least one vertex in $S$. We consider first the case in which the subhypergraph of $G$ induced by $S$ is connected. 
  By the definition of $d_c(G)$, there are still vertices not belonging to ${\cal N}(S)$, i.e.,\ $\overline{{\cal N}(S)} \neq \emptyset$. Let $\tau$ be an arbitrary $G$ state, i.e.,\
    \begin{equation}
        \tau = \otimes_{e\in E(G)} \gamma_e,
    \end{equation}
    where the $\{\gamma_e\}$ are the source states in the network $G$. Any LOCC operation $\Lambda$ as above can be written as
    \begin{equation}
        \Lambda(\cdot) = \Big(\sum\nolimits_{i_v, v\in V} \Big[\bigotimes\nolimits_{u\in S} A^{(u)}_{i_S}\Big] \otimes \Big[\bigotimes\nolimits_{u\in \bar{S}} B^{(u)}_{i_S, i_u} \Big]\Big) \cdot \Big(\sum\nolimits_{i_v, v\in V} \Big[\bigotimes\nolimits_{u\in S} A^{(u)\dagger}_{i_S}\Big] \otimes \Big[\bigotimes\nolimits_{u\in \bar{S}} B^{(u)\dagger}_{i_S, i_u} \Big]\Big),
    \end{equation}
    where $i_S = \{i_v\}_{v\in S}$ encodes the information sent around in the LOCC protocol, the superscript $(u)$ in $A^{(u)}_{i_S}$ and $B^{(u)}_{i_S, i_u}$ indicates the vertex being acted on, and
    \begin{align}
        &\sum\nolimits_{i_S} \Big[\bigotimes\nolimits_{u\in S} A^{(u)\dagger}_{i_S}\Big] \Big[\bigotimes\nolimits_{u\in S} A^{(u)}_{i_S}\Big] = \id,\\
        &\sum\nolimits_{i_u} B^{(u)\dagger}_{i_S, i_u} B^{(u)}_{i_S, i_u} = \id, \forall i_S.
    \end{align}
    For convenience, let $A_{i_S} = \bigotimes\nolimits_{u\in S} A^{(u)}_{i_S}$, $ B_{i_S, i_T} = \otimes_{u\in T} B^{(u)}_{i_S, i_u}$ and $S' = S\cup \overline{{\cal N}(S)}$.
    Then we have 
    \begin{align}
        &\tr_{{\cal N}(S)\setminus S} \Lambda(\tau) \nonumber \\
        = & \sum_{i_v, v\in V} \Big(A_{i_S}\otimes B_{i_S, i_{ \overline{{\cal N}(S)}}}\Big) \tr_{{\cal N}(S)\setminus S}[(\id_{S'} \otimes B_{i_S,i_{{\cal N}(S)\setminus S}}) \tau (\id_{S'}\otimes B_{i_S,i_{{\cal N}(S)\setminus S}}^\dagger )] \Big(A_{i_S}\otimes B_{i_S, i_{ \overline{{\cal N}(S)}}}\Big)^\dagger\\
        =& \sum_{i_v, v\in S'} \Big(A_{i_S}\otimes B_{i_S, i_{ \overline{{\cal N}(S)}}}\Big) \tr_{{\cal N}(S)\setminus S}(\tau) \Big(A_{i_S}\otimes B_{i_S, i_{ \overline{{\cal N}(S)}}}\Big)^\dagger\\
        =& \sum_{i_v, v\in S' } \Big(A_{i_S}\otimes B_{i_S, i_{ \overline{{\cal N}(S)}}}\Big) \tr_{{\cal N}(S)\setminus S}(\otimes_{e\in E(G)} \gamma_e) \Big(A_{i_S}\otimes B_{i_S, i_{ \overline{{\cal N}(S)}}}\Big)^\dagger.
    \end{align}
    By definition of ${\cal N}(S)$, there is no edge $e\in E(G)$ such that $e\cap S \neq \emptyset$ and $e\cap \overline{{\cal N}(S)} \neq \emptyset$ at the same time. Hence, $\tr_{{\cal N}(S)\setminus S}(\otimes_{e\in E(G)} \gamma_e)$ is a separable state in the bipartition $S | \overline{{\cal N}(S)}$, which is still so after the separable map
    \begin{equation}
        \sum_{i_v, v\in S' } \Big(A_{i_S}\otimes B_{i_S, i_{ \overline{{\cal N}(S)}}}\Big) \cdot \Big(A_{i_S}\otimes B_{i_S, i_{ \overline{{\cal N}(S)}}}\Big)^\dagger.
    \end{equation}
    Consequently, the reduced state of any state $\Lambda(\tau)$ remains separable in the bipartition $S | \overline{{\cal N}(S)}$ and, since the set of separable states is closed, the same applies to any state that could be approximated to arbitrary accuracy by a sequence of such protocols. On the other hand, the assumption that all possible bipartite reduced states of $\rho$ for non-adjacent vertices are entangled implies that all reduced states of $\rho$ are entangled in all bipartitions that split at least one pair of non-adjacent vertices. As this is the case for $S | \overline{{\cal N}(S)}$, this entails that $\epsilon_r(\rho|G)>d_c(G)-1$. Given that $\rho$ is pure, this leads to $\ecom(\rho|G)>d_c(G)-1$, finishing the proof for the case in which the subnetwork of $G$ induced by $S$ is connected. An example with a grid graph is provided in Fig.~\ref{fig:exp12}.

    Similarly, when the subnetwork of $G$ induced by $S$ is not connected, the reduced state 
    \begin{align}
      \tr_{\bar{S}} \Lambda(\tau) = \sum_{i_v, v\in S} A_{i_S} \tr_{\bar{S}} (\tau) A_{i_S}^\dagger
    \end{align}
    should be separable in some bipartition since $\tr_{\bar{S}}(\tau)$ must be separable in some bipartition and the map
     $ \sum_{i_v, v\in S} A_{i_S} \cdot A_{i_S}^\dagger $
     is a separable map. Thus, the same argument as above leads us to the desired claim in this case as well.
\end{proof}

\begin{figure}[th]
  \centering
\subfloat[\label{subfig:exp1}]{%
  \includegraphics[width=0.4\columnwidth]{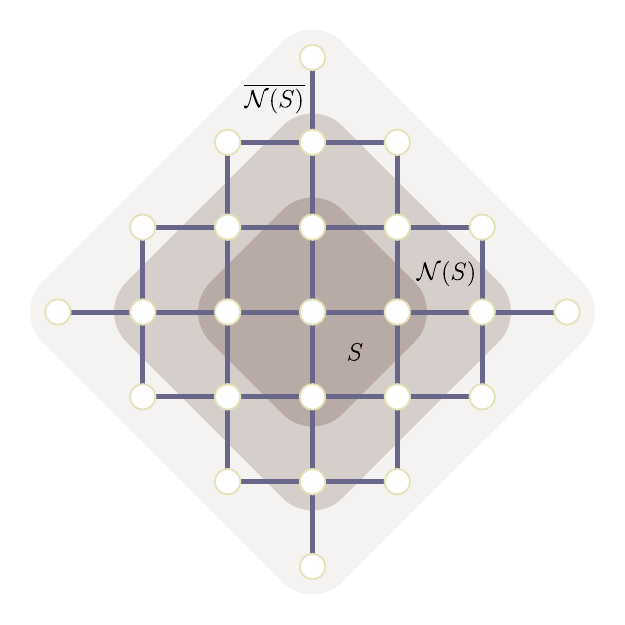}%
}\hspace{5em}
\subfloat[\label{subfig:exp2}]{%
  \includegraphics[width=0.4\columnwidth]{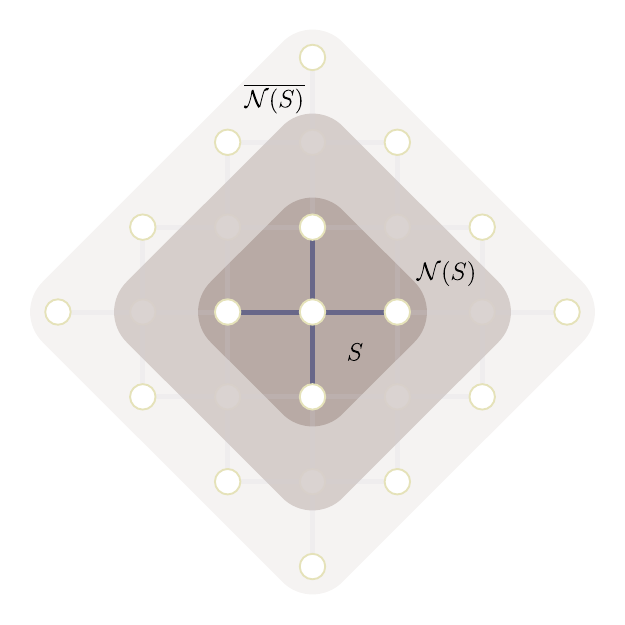}%
}
\caption{A network represented by a grid graph in (a) and the reduced subnetwork by excluding the vertices in ${\cal N}(S)\setminus S$ in (b).
}\label{fig:exp12}
\end{figure}

Notice that {for a hypergraph $G$ in which every pair of vertices is adjacent} the fact that $\ecom(\rho|G)=d_c(G)=1$ for every pure $G$-entangled state is obvious. Thus, Observation 3 shows that the upper bound in Observation 2 is tight in the sense that there always exists a large class of pure states that saturates it for any network $G$. Notice that a particular example of states that meet the hypothesis of this observation {for any hypergraph of $n$ vertices in which there is a non-adjacent pair} is the $n$-qubit W state~\cite{dur2000w}. 

However, for a given network and a given state, strict inequality may hold. Take, for instance, the state $|\psi\rangle =(|000\rangle + |111\rangle)\otimes|0\rangle/\sqrt{2}$, which must be $G$-entangled for any graph $G$ since there can at most be bipartite sources, and consider for example $G=C_4$ (i.e.\ a cycle with 4 vertices). Hence, $\efr(|\psi\rangle\langle \psi| | C_4) = 1$. However, the state $|\psi\rangle$ can be prepared from $C_4$-states with only one round of LOCC, i.e.,\ $\ecom(|\psi\rangle\langle\psi| | C_4) =\efr(|\psi\rangle\langle \psi| | C_4)= 1 < 2 = d_c(C_4)\efr(|\psi\rangle\langle \psi| | C_4)$. The procedure to generate $|\psi\rangle$ from a network state with round complexity $1$ is as follows:
\begin{align}
    & |\Psi^+\rangle_{12} \otimes |\Psi^+\rangle_{23} \otimes |0\rangle_4\\
    \overset{C_2}{\to} & (|0000\rangle + |0011\rangle + |1110\rangle + |1101\rangle)\otimes |0\rangle_4/2\\
    \overset{M_2}{\to} & 0: (|000\rangle + |111\rangle)\otimes |0\rangle_4/\sqrt{2},\ 1: (|001\rangle + |110\rangle)\otimes|0\rangle_4/\sqrt{2}\\
    \overset{U_3}{\to} & (|000\rangle + |111\rangle)\otimes |0\rangle_4/\sqrt{2},
\end{align}
where $|\Psi^{+}\rangle = (|00\rangle + |11\rangle)/\sqrt{2}$,  $C_2$ is the CNOT gate controlled by the first qubit on the second vertex, $M_2$ is the measurement $\{|0\rangle\langle 0|, |1\rangle\langle 1|\}$  on the second qubit on the second vertex, and $U_3$ is the unitary depending on the measurement outcome of $M_2$, i.e., $U_3 = \id$ if the outcome is $0$ and $U_3 = X$ if the outcome is $1$.

\section{Proof of Observation 4}

Before addressing the main goal of this section, let us point out that, on the analogy of the example we have just considered for $C_4$, we can prove strict inequality with the upper bound and attainability of the lower bound in Observation 2 for particular states and networks in this case as well. To wit, $\ecost(|\phi\rangle \langle \phi| |C_5) = 1 < r_h(C_5) = 2$, where $|\phi\rangle = (|0000\rangle+|1111\rangle)\otimes |0\rangle/\sqrt{2}$ and $C_5$ is the cycle graph with 5 vertices. The argument follows the same logic as before using that the state $(|0000\rangle+|1111\rangle)/\sqrt{2}$ can be prepared in the line graph $L_4$ using only one step as $r_h(L_4) = 1$.
{We denote by $L_n$ the line graph of $n$ vertices, where $V=\{1,2,\ldots,n\}$ and $E=\{\{i,i+1\}:i=1,2,\ldots,n-1\}$.}

We move on now to prove Observation 4. In order to do this, we will derive {firstly} several lemmas. In what follows {we use the term \textit{action} (by party $i$) to refer to a} one-round LOCC protocol in which party $i$ implements a measurement, sends classical information to some non-empty subset of its neighbours and all parties implement then corrections accordingly. 
{A step then can be viewed as a concatenation of actions fulfilling the rule that we have given in the main text.} Notice that it follows from the definition that the order in which each acting party in a step acts can be chosen at will. 

\begin{lemma}\label{l1}
Given a graph $G$, if party $v$ acts, $w\in N(v)$, $w',w''\notin N(v)$, and the input state has a marginal for $vww'w''$ that is separable in the bipartition $vw|w'w''$, then the output state retains this property.
\end{lemma}

\begin{proof}
Let $\emptyset\neq S\subseteq N(v)$ be the set of neighbours of $v$ who receive classical information from $v$. The map $\Lambda$ corresponding to this action has Kraus operators of the form
\begin{equation}
\Big\{A_{i_v}^{(v)}\otimes\bigotimes_{u\in S}A_{i_vi_u}^{(u)}\otimes\bigotimes_{u\notin S}B_{i_u}^{(u)}\Big\},
\end{equation}
where
\begin{equation}
\sum_{i_v}(A_{i_v}^{(v)})^\dagger A_{i_v}^{(v)}=\id,\quad \sum_{i_u}(B_{i_u}^{(u)})^\dagger B_{i_u}^{(u)}=\id\quad\forall u\notin S,
\end{equation}
and for every $u\in S$ and every $i_v$
\begin{equation}\label{conditionLambda}
\sum_{i_u}(A_{i_vi_u}^{(u)})^\dagger A_{i_vi_u}^{(u)}=\id.
\end{equation}
Let $\tau$ denote the input state to the protocol. If we assume that $w\in S$, then it holds that
\begin{align}
&(\Lambda(\tau))_{vww'w''}=\nonumber\\&\sum_{i_v,i_w,i_{w'},i_{w''}}(A_{i_v}^{(v)}\otimes A_{i_vi_w}^{(w)}\otimes B_{i_{w'}}^{(w')}\otimes B_{i_{w''}}^{(w'')})\tau_{vww'w''}(A_{i_v}^{(v)}\otimes A_{i_vi_w}^{(w)}\otimes B_{i_{w'}}^{(w')}\otimes B_{i_{w''}}^{(w'')})^\dagger,
\end{align}
and the result follows. If, on the other hand, $w\notin S$ then
\begin{equation}
(\Lambda(\tau))_{vww'w''}=\sum_{i_v,i_w,i_{w'},i_{w''}}(A_{i_v}^{(v)}\otimes B_{i_w}^{(w)}\otimes B_{i_{w'}}^{(w')}\otimes B_{i_{w''}}^{(w'')})\tau_{vww'w''}(A_{i_v}^{(v)}\otimes B_{i_w}^{(w)}\otimes B_{i_{w'}}^{(w')}\otimes B_{i_{w''}}^{(w'')})^\dagger,
\end{equation}
and the result follows as well.
\end{proof}

\begin{lemma} \label{l2}
Given a graph $G$, if party $v$ acts, $w,w'\notin N(v)$, and the input state has a marginal for $ww'$ that is separable, then the output state retains this property.
\end{lemma}
\begin{proof}
With the same notation as in the previous lemma, we now obtain that
\begin{equation}
(\Lambda(\tau))_{ww'}=\sum_{i_w,i_{w'}}(B_{i_w}^{(w)}\otimes B_{i_{w'}}^{(w')})\tau_{ww'}(B_{i_w}^{(w)}\otimes B_{i_{w'}}^{(w')})^\dagger,
\end{equation}
and the result follows.
\end{proof}

\begin{lemma}\label{l3}
Any state $\rho$ such that $r_h(\rho|L_n)\leq\lceil n/2\rceil-2$ must have the property that the marginal $\rho_{1n}$ is separable.
\end{lemma}
\begin{proof} Notice that throughout this proof we can ignore shared randomness as separability is preserved by its use. Moreover, since the set of separable states is closed, it suffices to establish the claim for exact LOCC protocols. The proof is by induction. We will first prove the base cases $L_3$ and $L_4$ and then we will show that the property is true for $L_{n+2}$ if it holds for $L_n$.

To see that the claim holds for $L_3$ we have to consider protocols with communication cost equal to zero. Our input state $\tau$ corresponding to bipartite sources in $L_3$ (i.e.,\ $\tau$ is an arbitrary $L_3$-state) always has the property that $\tau_{13}$ is separable and the output state $\rho$ obtainable from a zero-communication protocol reads
\begin{equation}
\rho=\sum_{i_1,i_2,i_3}(A_{i_1}\otimes B_{i_2}\otimes C_{i_3})\tau(A_{i_1}\otimes B_{i_2}\otimes C_{i_3})^\dagger,
\end{equation}
and, hence,
\begin{equation}
\rho_{13}=\sum_{i_1,i_3}(A_{i_1}\otimes C_{i_3})\tau_{13}(A_{i_1}\otimes C_{i_3})^\dagger,
\end{equation}
obtaining the desired result.

That the result is true as well for $L_4$ readily follows from the case above. Indeed, we can regard $L_4$ as $L_3$ with parties 3 and 4 joint together. Therefore, in a zero-cost protocol $\rho_{134}$ must be separable in the bipartition $1|34$ and, consequently, $\rho_{14}$ has to be separable.

We assume now the claim to be true for $L_n$. We can regard $L_{n+2}$ as $L_n$ if we join together parties 1 and 2 and parties $n+1$ and $n+2$. Hence, by the induction hypothesis, any state $\rho'$ obtained after
$$\left\lceil\frac{n+2}{2}\right\rceil-3=\left\lceil\frac{n}{2}\right\rceil-2$$
steps must be such that $\rho'_{1,2,n+1,n+2}$ is separable in $1,2|n+1,n+2$. Now, we have one more step to reach $\rho$ from $\rho'$ consuming the overall $\lceil (n+2)/2\rceil-2$ steps. It can be that in this last step party 1 acts or not. If party 1 acts (and must therefore send classical information to 2), by Lemma \ref{l1} after this action the reduced output state remains separable in $1,2|n+1,n+2$ (and, consequently, in $1|n+2$). This property will be preserved within this last step by any potential action by $n+1$ or $n+2$ due to Lemma \ref{l1}. Finally, any potential actions by $\{3,4,\ldots,n\}$ will preserve separability of the reduced output state in $1|n+2$ by Lemma \ref{l2}. If party 1 does not act, independently of whether party 2 acts or not the reduced output state remains separable in $1,2|n+1,n+2$ at this point because of Lemma \ref{l1}. By the same argument as in the previous case the final reduced output state will be separable in $1|n+2$ after the potential actions of $\{3,4,\ldots,n+2\}$ within this last step.
\end{proof}  

\begin{lemma}\label{l4}
Any $n$-partite pure state $\rho$ such that the marginal $\rho_{1n}$ is entangled fulfills
$$\mathcal{E}_c(\rho|L_n)=r_h(L_n)=\left\lceil\frac{n}{2}\right\rceil-1.$$
\end{lemma}

\begin{proof}
It follows immediately from Lemma \ref{l3} that in this case $r_h(\rho|L_n)\geq\lceil n/2\rceil-1=r_h(L_n)$. Since $\rho$ is pure, this implies that $\ecost(\rho|L_n)\geq r_h(L_n)$, while Observation 2 tells us that $\ecost(\rho|L_n)\leq r_h(L_n)$.
\end{proof}

We are now in the position to prove Observation 4. \\
\noindent\textbf{Observation 4.}\textit{
 {Let $T$ be an arbitrary tree graph and let $d$ be its diameter, i.e., $d=\max_{u,v\in T} dis(u,v)$, and $\rho$ a pure state. If there exists at least one pair of parties at distance $d$ in $T$ for which the corresponding bipartite reduced state of $\rho$ is entangled, then  $\mathcal{E}_c(\rho|T)=r_h(T)$.}
    }
\begin{proof}
{Let $v_1$ and $v_{d+1}$ be an arbitrary pair of vertices at distance $d$ in $T$ such that the reduced state $\rho_{v_1v_{d+1}}$ is entangled and let $V'=\{v_i\}_{i=1}^{d+1}$ be the set of vertices in the path that connects them. Let moreover} $V'_i$ be the set of vertices not in $V'$ that branch out of $v_i$ for each $i$ (notice that $V'_1$ and $V'_{d+1}$ must be empty). Then,
\begin{equation}
    r_h(T)=\left\lceil\frac{d+1}{2}\right\rceil-1.
\end{equation}
Now, suppose that for each $i$ we allowed all parties corresponding to vertices in $V'_i$ to act jointly together with $v_i$. Then, the new network that we effectively obtain is $L_{d+1}$ (where the first vertex corresponds only to party $v_1$ and the last only to $v_{d+1}$). Thus, Lemma \ref{l4} {and our assumption on $\rho$ impose} that $\ecost(\rho|T)\geq r_h(T)$ and the result follows.
\end{proof}

\section{Proof of Observation 5}

In order to fulfill the goal of this section, we first quote two useful results from~\cite{hansenne2022symmetries}.
\begin{lemma}\label{lm:antisymm}
   Any $n$-partite anti-symmetric state or entangled symmetric state is $k$-network-entangled for $k<n$.
\end{lemma}
\begin{lemma}\label{lm:dec}
   Any state in a convex decomposition of an anti-symmetric state (a symmetric state) is still anti-symmetric (symmetric).
\end{lemma}

Now we can prove Observation 5.\\
\noindent\textbf{Observation 5.}\textit{
For any anti-symmetric state $\rho$, $\efr(\rho|G) = 1$. For any symmetric state $\rho$, $\efr(\rho|G)$ does not depend on the network topology $G$, {\color{black}i.e. the network-entanglement weight of $\rho$ is the entanglement weight of $\rho$.}
    }
\begin{proof}
    For a given network $G$, let $\rho$ denote the target state and $p = \efr(\rho|G)$.
    Then there exists a decomposition $\rho = (1-p)\sigma + p\tau$, where $\sigma\in \scq(G)$ and $\tau$ is an arbitrary state. 
    We firstly consider the case that $\rho$ is anti-symmetric.
    From Lemma~\ref{lm:dec}, we know that $\sigma$ is still anti-symmetric. Then Lemma~\ref{lm:antisymm} leads to an contradiction if $(1-p)\neq 0$. Thus, $\efr(\rho|G)=p$ can only be $1$.

    If $\rho$ is symmetric, similarly, $\sigma$ is still symmetric. Then Lemma~\ref{lm:antisymm} implies that $\sigma$ can only be separable. Thus, $\sigma$ can be prepared in any quantum network. Consequently, $\efr(\rho|G)=p$ is independent of the network topology $G$.
\end{proof}
This result implies that for symmetric states the network entanglement weight boils down to the entanglement weight, i.e.,\ the minimization over  $\scq(G)$ can be replaced by a minimization over the set of fully separable states, which is much better characterized. 

\section{Estimation based on covariance matrix}
For a given state $\rho$,
denote $p=\efr(\rho)$ and $\rho = (1-p) \sum_i p_i \sigma_i + p\tau$ with $\sigma_i \in \siq$, $\Gamma(\rho)$, $\Gamma(\rho_i)$'s and $\Gamma(\tau)$ 
the covariance matrices of $\rho$, $\rho_i$'s and $\tau$ with the $\pm 1$-measurement $M_i$ on the $i$th party.
We introduce the matrix $T$ as
\begin{equation}
    T = \Gamma(\rho) - \Big[\sum_i (1-p) p_i \Gamma(\rho_i) + p\Gamma(\tau)\Big],
\end{equation}
which is also semi-definite and $|T_{ij}| \le \beta(\rho) = 2\sqrt{1-\xi^2}$ with $\xi = \tr(\rho^2)$~\cite{xu2023characterizing}. 	Besides, it is known that~\cite{xu2023characterizing} 
\begin{align}
&\sum_{jl} \Gamma_{jl}(\rho_i) \le k\tr(\Gamma(\rho_i)),\\
&\sum_{jl}\Gamma_{jl}(\tau) \le n\tr(\Gamma(\tau)),\\
&\sum_{jl} T_{jl} \le n \tr(T).
\end{align}
Thus, we know that
\begin{align}
    \sum\nolimits_{jl} \Gamma_{jl}(\rho) 
    &\le k\sum_i (1-p)p_i \tr(\Gamma(\rho_i)) + np\tr(\Gamma(\tau)) + n\tr(T)  \nonumber\\
    &=  k\tr(\Gamma(\rho)) + (n-k) [p\tr(\Gamma(\tau)) + \tr(T)]  \nonumber\\
    &\le  k\tr(\Gamma(\rho)) + n(n-k) [p + \beta(\rho)].
\end{align}
In the second line, we have made use of the fact that {\color{black}$(1-p)\sum_i p_i\tr(\Gamma(\rho_i)) = \tr(\Gamma(\rho)) - p\tr(\Gamma(\tau)) - \tr(T)$.} The third line holds because each diagonal term in $\Gamma(\tau)$ is no more than $1$ and each diagonal term in $T$ is no more than $\beta(\rho)$. As a result, for all $k$-networks, we have that
\begin{align}\label{eq:cov_lower}
    \efr(\rho) \ge \frac{\omega(\rho)}{n(n-k)} - \beta(\rho),
\end{align}
where 
$\omega(\rho) = \sum_{ij} \Gamma_{ij}(\rho) - k \tr(\Gamma(\rho))$, 
and $\beta(\rho) = 2\sqrt{1-[\tr(\rho^2)]^2}$. 
The estimation in Eq.~\eqref{eq:cov_lower} is also tight for the GHZ state $|\ghz\rangle$ with only the Pauli measurement $Z$, i.e., $M_i=Z$ for all $i$.

For the $n$-partite noisy GHZ state $\rho = p|\ghz\rangle\langle \ghz| + (1-p)\id /2^n$  with local dimension $2$ and the measurements to be $Z$ for all the parties, we have
\begin{equation}
  \omega(\rho) = pn(n-k),\ \beta(\rho) = 2\sqrt{ 1-\left[p^2+\frac{1-p^2}{d^n}\right]^2  }.
\end{equation}

\section{Estimation based on a see-saw method}
All the estimations given in the main text are from below, here we use a see-saw method and semi-definite programming to estimate $\efr(\rho)$ from above. Nevertheless, this method works in practice only for small networks due its computational cost.

By definition, $\efr(\rho) = \min_{\sigma\in \scq(G), \tau} p$ such that $\rho = (1-p)\sigma + p\tau$. Denote $\tilde{\sigma} = (1-p)\sigma$, the definition of $\efr(\rho)$ can be reformulated as follows:
\begin{align}\label{eq:see-saw-sdp}
    &\min \tr(\rho-\tilde{\sigma})\nonumber\\
    s.t.\ & \rho \ge \tilde{\sigma} = \sum\nolimits_{\lambda} p_\lambda \left(\otimes_{i\in V} C_i^{(\lambda)}\right)\left[\otimes_{e\in E} \rho_e^{(\lambda)}\right], \nonumber\\
    & \sum\nolimits_\lambda p_\lambda \le 1,\ p_{\lambda} \ge 0,\ 
 C_i^{(\lambda)} \text{ is a channel},\  \rho_e^{(\lambda)} \ge 0.
\end{align}
In principle, the size of $\{p_\lambda\}$ could be infinite. If we specify the size of $\{p_\lambda\}$ to be a finite number, the resulting optimal value is an upper bound of $\efr(\rho)$. This leads to a see-saw method. 

\begin{figure}[h]
    \centering
    \includegraphics[width=0.6\textwidth]{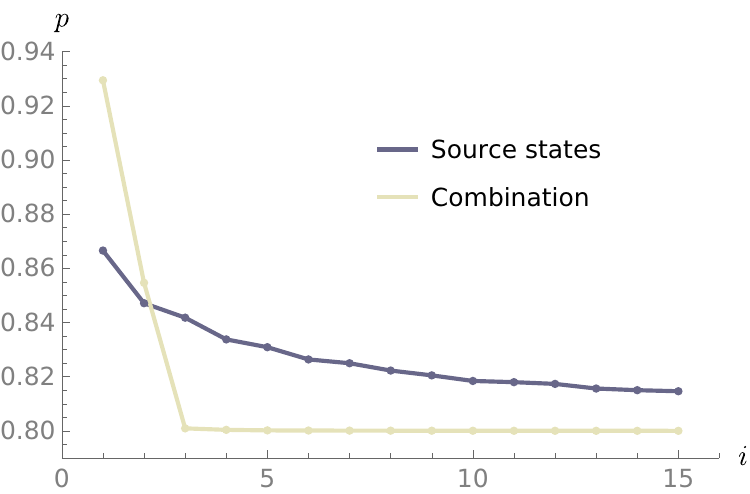}
    \caption{Estimation $p$ of $\efr(\rho)$ with the see-saw method for the tripartite state $\rho = 0.8 {\rm GHZ} + 0.2 \id/64$ whose local dimension is $4$, the quantum network is the triangle network. For the combination method (blue line), we generate $10\times i$ network states randomly and optimize over the non-negative combination for $i=1, \ldots, 15$. For the method with optimization over source states, we generate the source states randomly, and then optimize one source state at the $i$-th step.}
    \label{fig:seesaw}
\end{figure}

Initially, we generate random assignments of $C_i^{(\lambda)}$'s and $\rho_e^{(\lambda)}$'s, then we optimize over $\{p_\lambda\}_{\lambda}$. At each of the following steps, we pick up one of $\{p_\lambda\}_{\lambda}$, $C_i^{(\lambda)}$'s and $\rho_e^{(\lambda)}$'s as the variable and fix the rest, then we optimize over this variable. This is always a semi-definite program. The solution at each step provides an upper bound of $\efr(\rho)$. The quality of this upper bound depends on the initial assignments, the size of $\{p_\lambda\}$ and the number of steps.

Here we take two simplified versions of the original see-saw method for illustration, where the state is the tripartite state $\rho = 0.8 {\rm GHZ} + 0.2 \id/64$ whose local dimension is $4$, and the quantum network is the triangle network. In the first version, we firstly generate a set of network states randomly, i.e., a set of local channels and source states. Then we optimize over the non-negative combination of those network states.
In the second version, we firstly generate the source states randomly, and then optimize each source state in an iterative way. In this simplified version, we take no local channels and  combinations. As illustrated in Fig.~\ref{fig:seesaw}, the second version converges to $0.8$ much faster and performs better.

\section{Witness and estimation based on witness}
\begin{lemma}[~\cite{lenard1972numerical}]
For two projectors $P$ and $Q$ sharing no common eigenvector, there is a unitary $U$ such that 
\begin{align}\label{eq:bd}
&	UPU^\dagger = \big[\bigoplus\nolimits_{i=1}^t P_i\big] \oplus D_P,\\
&	UQU^\dagger = \big[\bigoplus\nolimits_{i=1}^t Q_i\big] \oplus D_Q,
\end{align}
where $D_P$ and $D_Q$ are {diagonal matrices whose entries are either $1$ or $0$}, 
\begin{align}
	P_i = \begin{bmatrix} 1 & 0\\ 0& 0 \end{bmatrix},\ Q_i=\begin{bmatrix} c_i^2&c_is_i\\ c_is_i&s_i^2 \end{bmatrix}
\end{align}
where $c_i$'s and $s_i$'s are positive real numbers and $c_i^2+s_i^2=1$.
\end{lemma}

\begin{lemma}

(Uncertainty relation for two projections) For two projections $P$ and $Q$ satisfying $PQP=\lambda P$ with $0<\lambda<1$, it holds
\begin{equation}
\langle Q\rangle\le f_\lambda(\langle P\rangle),\quad f_\lambda(x)=1-\big(\sqrt{\lambda(1-x)}-\sqrt{(1-\lambda)x}\big)^2,
\end{equation}
\end{lemma}
\begin{proof}
	Without loss of generality, we assume that $P$ and $Q$ are already in the block diagonal form in Eq.~\eqref{eq:bd}. 
The condition that $PQP=\lambda P$ implies that $P_iQ_iP_i = \lambda P_i$ and $D_PD_QD_P=\lambda D_P$.
In the case that $D_P$ and $D_Q$ are not empty, $\lambda$ can only be either $1$ or $0$, which contradicts with the assumption that $\lambda \in (0,1)$. Hence, both of $D_P$ and $D_Q$ are empty.

Then for any state $\rho$, we have 
\begin{equation}
	\langle P\rangle = \tr(P\rho) = \sum_{i=1}^t \tr(P_i \rho_i),\ 
	\langle Q\rangle = \tr(Q\rho) = \sum_{i=1}^t \tr(Q_i \rho_i),
\end{equation}
where $\rho_i$ is the $i$-th diagonal $2\times 2$ block of $\rho$.

Since $f_{\lambda}(x)$ is a convex function for any $\lambda \in (0,1)$, we only need to prove that $\langle Q_i\rangle \le f_{\lambda}(\langle P_i\rangle)$.

Denote $x=\langle P_i\rangle$ and the state is $\rho_i$, then the form of $P_i$ implies that 
\begin{equation}
	\rho_i = \begin{bmatrix} x & y\\ y^\dagger & 1-x \end{bmatrix},
\end{equation}
where $|y| \le \sqrt{x(1-x)}$.
Since $Q_i$ is also a real matrix, we can assume $y$ to be real without changing $\langle P_i\rangle$ and $\langle Q_i\rangle$.

It is direct to see that 
\begin{align}
	\langle Q_i\rangle &= \lambda x + (1-\lambda)(1-x) + 2\sqrt{\lambda(1-\lambda)} y\\
			   &\le \lambda x + (1-\lambda)(1-x) + 2\sqrt{\lambda(1-\lambda)} \sqrt{x(1-x)}\\
			   &= 1- \big(\sqrt{x(1-\lambda)} - \sqrt{(1-x)\lambda}\big)^2,
\end{align}
with $x=\langle P_i\rangle$.
\end{proof}
For a unitary operator $g$, denote $\ty g$ the projector into the eigenspace with eigenvalue $1$ of $g$, and $\langle \ty g\rangle_{\rho}$ the corresponding mean value with the state $\rho$. {In the case that $1$ is not an eigenvalue of $g$, the projector $\ty g = 0$.}

{\color{black}In the following we use $A\succeq B$ for Hermitian matrices $A$ and $B$ to denote that $A-B$ is positive semidefinite.}
\begin{lemma}[~\cite{wang2024quantum}]
For two commuting unitaries $g$ and $h$, it holds
$$ \ty {gh}{\color{black}\succeq} \ty{g}\ty{h}{\color{black}\succeq} \ty g+\ty{h}-{\color{black}\id}.$$
\end{lemma}
\begin{lemma}
If the unitary $g$ satisfying $g^d={\color{black}\id}$ for some integer $d$, we have
\begin{equation}
\lceil g\rceil=\frac1d\sum_{k=0}^{d-1} g^k.
\end{equation}
\end{lemma}
\begin{proof}
{As $g^d={\color{black}\id}$ the eigenvalues of $g$ are $\omega^j$ with $\omega=e^{i\frac{2\pi}d}$ and $j=0,\ldots, d-1$. Let $|g_j\rangle$ be the eigenstate of $g$ corresponding to eigenvalue $\omega^j$ and it holds
\begin{equation}\frac1d\sum_{k=0}^{d-1}g^k=\frac1d\sum_{j,k=0}^{d-1}\omega^{jk}|g_j\rangle\langle g_j|=|g_0\rangle\langle g_0|=\lceil g\rceil.
\end{equation}
}


\end{proof}
{To proceed, we firstly recall the definition of $d$-dimensional Pauli matrices, i.e.,
\begin{equation}
    Z=\sum_{k=0}^{d-1}\omega^{k}|k\rangle\langle k|,\quad X=\sum_{k=0}^{d-1}|k\rangle\langle k\oplus 1|.
\end{equation}
Then a $d$-dimensional Pauli string for a $n$-partite system means $\otimes_{i=1}^n M_i$ where $M_i$ is either $X, Z$ or identity $\id$. We also denote $Z_j$ (or $X_j$) the corresponding Pauli string with $Z$ (or $X$) acting on the $j$-th party only for convenience. 
}
\begin{theorem}
  For any two {non-commuting operators $g$ and $h$ as strings of $d$-dimensional Pauli matrices} such that $gh=\omega hg$ where $\omega^d = 1$, we have $\ty g \ty h \ty g = \ty g/d$ and consequently $\ty{h}_\rho \le f_{1/d}(\ty{g}_{\rho})$ for any state $\rho$.
\end{theorem}
\begin{proof}
  Direct calculation shows that
  \begin{align}
    \ty g \ty h \ty g = \frac{1}{d^3} \sum_{k,k',j=0}^{d-1} g^k h^j g^{k'} = \frac{1}{d^3} \sum_{k,k',j=0}^{d-1} \omega^{kj} h^j  g^{k+k'} = \frac{1}{d^2} \sum_{k,k'=0}^{d-1} \ty{\omega^k h} g^{k+k'} = \frac{1}{d} \sum_{k=0}^{d-1} \ty{\omega^k h} \ty{g} = \frac{1}{d}\ty{g},
  \end{align}
  where we have made use of the facts that $\sum_{k'=0}^{d-1} g^{k+k'} = \sum_{k'=0}^{d-1} g^{k'} = d\ty g$ {for all $k$}, and $\sum_{k=0}^{d-1} \ty{\omega^k h} = \id$.
\end{proof}

\begin{theorem}
  For a given $n$-partite state $\rho \in \scq(G_{n,k})$ with local dimension $d$, where {$G_{n,k}$} is a hypergraph with $n$ vertices and all hyperedges containing exactly $k$ vertices, then the following uncertainty relation holds,
  \begin{equation}\label{eq:un2}
\frac{2dk}{n(n-1)}\sum_{j>l}(1-\ty {Z_jZ_l^\dagger}_\rho)\ge \left(\sqrt{(d-1)\ty{X}_\rho}-\sqrt{1-\ty{X}_\rho}\right)^2,
\end{equation}
{where $j, l$ in the sum run over the set of vertices of $G_{n,k}$ and $X=\prod_{j=1}^n X_j$}.

\end{theorem}
\begin{proof}
  Firstly, we prove that
  \begin{equation}\label{eq:un1}
  {d\left[ (1-\ty {Z_kZ_{n}^\dagger}_\rho) + \sum_{j=1}^{k-1}(1-\ty {Z_jZ_{j+1}^\dagger}_\rho)\right]}\ge \left(\sqrt{(d-1)\ty{X}_\rho}-\sqrt{1-\ty{X}_\rho}\right)^2.
\end{equation}
Then the symmetry of the hypergraph $G_{n,k}$ leads to the final conclusion by replacing vertices $1, \ldots, k$ in Eq.~\eqref{eq:un1} with any $k$ different vertices and taking averge of all the obtained inequalities.

Here we use the same symbols for the hypergraph and the corresponding state without confusion. 
To prove the inequality in Eq.~\eqref{eq:un1}, we {\color{black}adopt the inflation technique~\cite{wolfe2019inflation,navascues2020genuine} and} consider two inflations of $G_{n,k}$ and correspondingly the state $\rho$. The first inflation is just simply two copies of $G_{n,k}$ with vertices labeled by $1,\ldots, n, 1', \ldots, n'$. Denote $\tau$ the corresponding state. The second inflation $\eta$ takes the same labels of vertices, and hyperedges twisted from the ones of hypergraph $\tau$, {\color{black} which is described in details as follows.}

Notice that, for any hyperedge $e$ of $\tau$ which contains vertices $1$ and $n$, it cannot contain all of $2, \ldots,k$ at the same time since the size of $e$ is $k$. Denote $e'$ the hyperedge in $\tau$ with vertices $\{v'|v\in e\}$, and $i$ the minimal number in $\{1,\ldots,k\}$ such that $i\in e$ and $i+1\not\in e$. Then the twisted hyperedge $\epsilon$ for $\eta$ is obtained from $e$ by replacing $1,\ldots,i$ with $1',\ldots,i'$ and keeping others. The twisted hyperedge $\epsilon'$ for $\eta$ is obtained from $e'$ by replacing $1',\ldots,i'$ with $1,\ldots,i$ and keeping others. For all the rest hyperedges of $\tau$, we simply import them also into hypergraph $\eta$ without twisting.

By construction, we have the marginal relations that $\tau_{i,i+1} = \eta_{i,i+1} = \rho_{i,i+1}$ for all $i=1,\ldots,k$, since those two vertices $i$ and $i+1$ have the isomorphic neighboring hyperedges in hypergraphs $\tau$, $\eta$ and $K_{k,n}$ for $\rho$. Similarly, $\tau_{k,n}=\eta_{k,n}$ and $\tau_{1,n'} = \eta_{1,n}$. Consequently,
\begin{align}
  \ty{Z_1Z_{n'}}_\tau = \ty{Z_1Z_n}_\eta,\ \ty{Z_iZ_{i+1}}_\eta = \ty{Z_iZ_{i+1}}_\rho,\  \ty{X}_\tau = \ty{X}_\rho.
\end{align}
Besides, for $i=1,\ldots,k-1$, we have
\begin{equation}
\ty{Z_iZ_n}_\eta \ge \ty{Z_iZ_{i+1}}_\eta + \ty{Z_{i+1}Z_n}_\eta -1 = \ty{Z_iZ_{i+1}}_\rho + \ty{Z_{i+1}Z_n}_\eta -1,
\end{equation}
which implies that
\begin{align}
   \ty{Z_1Z_{n'}}_{\tau} = \ty{Z_1Z_n}_\eta \ge  \sum_{i=1}^{k-1} \ty{Z_iZ_{i+1}}_\rho + \ty{Z_kZ_n}_\rho - (k-1).
\end{align}
On the other hand, we know that
\begin{equation}
  \ty{Z_1Z_{n'}}_{\tau} \le f_{\frac1d} (\ty{X}_{\tau}) = f_{\frac1d} (\ty{X}_{\rho}).
\end{equation}
By combining the above inequalities, we obtain
\begin{equation}
  f_{\frac1d} (\ty{X}_{\rho}) \ge \sum_{i=1}^{k-1} \ty{Z_iZ_{i+1}}_\rho + \ty{Z_kZ_n}_\rho - (k-1),
\end{equation}
which is just the one in Eq.~\eqref{eq:un1} by rearranging terms in the inequality.
\end{proof}

\begin{theorem}
  For a given $n$-partite state $\rho \in \scq(G_{n,k})$ with local dimension $d$, where $G_{n,k}$ is a hypergraph with $n$ vertices and all hyperedges containing exactly $k$ vertices, denote $F = \langle \ghz|\rho|\ghz\rangle$ where $|\ghz\rangle = \sum_{i=0}^{d-1} |i\ldots i\rangle /\sqrt{d}$, then
  \begin{equation}\label{eq:fid}
  F \le \frac{(1+\sqrt{dk})^2}{(d-1)+(1+\sqrt{dk})^2} \to \left.\frac k{k+1}\right|_{d\to\infty}.
\end{equation}
\end{theorem}
\begin{proof}
In the case that $F\le 1/d$, the inequality in Eq.~\eqref{eq:fid} always holds. Then we only need to focus on the case that $F\ge 1/d$. 
  Since $Z_jZ_{l}$'s and $X$ are stabilizers of the GHZ state $|\ghz\rangle$, it holds that 
  \begin{equation}
    \ty{Z_jZ_l}_\rho \ge F,\ \ty{X}_\rho \ge F.
  \end{equation}
  Notice that the right hand of Eq.~\eqref{eq:un2} is monotonically increasing for $\ty{X}_{\rho}\ge F \ge 1/d$, this inequality implies that
  \begin{equation}
    dk(1-F) \ge \big(\sqrt{(d-1)F} - \sqrt{1-F}\big)^2.
  \end{equation}
 Equivalently,
 \begin{equation}
   (1+\sqrt{dk})\sqrt{1-F} \ge \sqrt{(d-1)F},
 \end{equation}
 which leads to
 \begin{equation}
   F \le \frac{(1+\sqrt{dk})^2}{(d-1)+(1+\sqrt{dk})^2},
 \end{equation}
 where the right hand decreases monotonically for $d$. Thus,
 \begin{equation}
   F \le \frac{(1+\sqrt{2k})^2}{1+(1+\sqrt{2k})^2}.
 \end{equation}
\end{proof}
Let us consider the noisy GHZ state
\begin{equation}
    \rho = p|\ghz\rangle\langle \ghz| + (1-p) \frac{\id}{d^n}.
\end{equation}
Since  $\langle \ghz|\rho|\ghz\rangle = p + \frac{(1-p)}{d^n}$ violate the inequality in Eq.~\eqref{eq:fid} when 
\begin{equation}
    p > \frac{(1+\sqrt{dk})^2 - (d-1)/(d^n-1)}{(1+\sqrt{dk})^2 + (d-1)},
\end{equation}
$\rho$ cannot be from a network with $k$-partite sources in this case.

Denote $    W_{k,d} = {k_d} \id - (k_d+1)|\ghz\rangle\langle \ghz|$ with $k_d = (1+\sqrt{dk})^2/(d-1)$, and $w_{k,d}(\rho) = \max \{0, - \tr(W_{k,d}\,\rho)\}$. Then for the noisy GHZ state with local dimension $d=2$, we have
\begin{equation}
  {\cal E}_w(\rho) \ge w_{k,2} = p + (1-p) \left[ \frac{k_2+1}{2^n} - k_2 \right].
\end{equation}
{
\color{black}
As shown in Fig.~\ref{fig:diff}, numerical results indicate that the upper bound $F$ in Eq.~\eqref{eq:fid} is always tighter than the one $F'$ from Ref.~\cite{navascues2020genuine}, which reads
\begin{equation}\label{eq:navasque}
    F' \leq \frac{d \left(3 - (k+1)(d+1) + (k+1)^2 d + 2 \sqrt{2 + (k+1)(d-1) - d}\right)}{1 + 4d - 2d(k+1) + (k+1)^2 d^2}.
\end{equation}
Besides, $F$ holds for any size $n$ of network which is no less than $k+1$, $F'$ was established only for the case that $n=k+1$.
}
\begin{figure}
    \centering
    \includegraphics[width=0.5\linewidth]{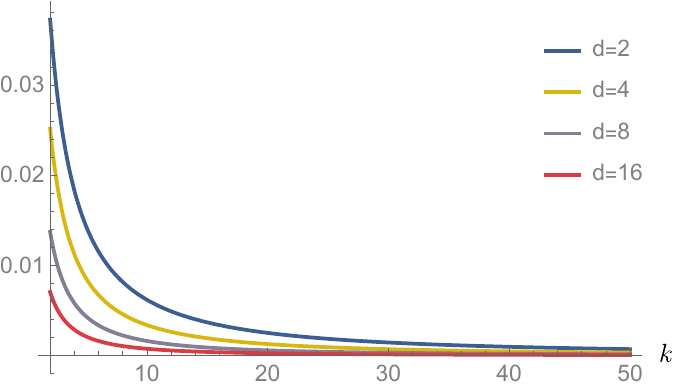}
    \caption{The difference $F'-F$ for $F'$ in Eq.~\eqref{eq:navasque} and $F$ in Eq.~\eqref{eq:fid}, which is always non-negative.}
    \label{fig:diff}
\end{figure}

\section{Other measures}
The formalism of quantum resource theories~\cite{chitambar2019quantum} can be readily applied to the network-entanglement scenario to obtain many different measures. Particularly, quantifiers based on the trace distance and fidelity have been often used in entanglement theory~\cite{horodecki2009quantum} and elsewhere. Their mathematical structure makes its estimation feasible and, in turn, they can be used to bound other measures~\cite{sun2024bounding}. 
Thus, we define
\begin{align}
    \ebu(\rho) &= \min_{\sigma \in \scq} 1-\sqrt{F}(\rho, \sigma),\ 
    \etr(\rho) = \min_{\sigma \in \scq} |\rho-\sigma|_{1}/2,\nonumber\\
    \betr(\rho) &= \min_{\{p_t,\rho_t\}, \sigma_t \in \siq} \sum\nolimits_t p_t |\rho_t - \sigma_t|_{1}/2,
\end{align}
where $|\cdot|_{1}$ is the trace norm, $F(\rho,\sigma) = |\sqrt{\rho} \sqrt{\sigma}|_{1}^2$, and $\siq(G)$ is the closure of the set of $G$ states post-processed without shared randomness. 
According to the quantum resource theory, $\ebu, \etr, \betr$ are also convex, LOSR monotonic and subadditive.

\begin{theorem} For any state $\rho$ and any network it holds that
    $\betr(\rho) \ge \etr(\rho)$.
\end{theorem}
\begin{proof}
By definition, there exists a decomposition $\{p_t, \rho_t\}$ of $\rho$ such that the equality in the following first line holds.
    \begin{align}
        2\betr(\rho) &= \min_{\sigma_t \in \siq} \sum_t p_t |\rho_t - \sigma_t|_{1} \ge \min_{\sigma_t \in \siq} |\rho - \sum_t p_t \sigma_t|_{1} \ge \min_{\sigma \in \scq} |\rho - \sigma|_{1} = 2\etr(\rho).
    \end{align}
\end{proof}
Hence, any lower bound of $\etr(\rho)$ is automatically a lower bound of $\betr(\rho)$. Similarly, the fact that $(|\rho - \sigma|_1/2)^2 \le 1-F(\rho,\sigma) \le 2(1-\sqrt{F(\rho,\sigma)})$, which implies that for any network
\begin{equation}
    \ebu(\rho) \ge {\etr^2(\rho)}/{2}.
\end{equation}
Hence, any lower bound of $\etr(\rho)$ induces automatically a lower bound on $\ebu(\rho)$ as well.

Those measures can also be estimated as follows.
\begin{theorem}
    For any state $\rho$ and any $k$-network it holds that $\etr(\rho) \ge w_k(\rho)$.
\end{theorem}
\begin{proof}
    \begin{align}
        2\etr(\rho) &= \min_{\sigma \in \scq} |\rho-\sigma|_1 = |\rho-\sigma_{\min}|_1 \ge |\tr[(\sigma_{\min} - \rho) (\id -2\operatorname{GHZ})]| = | \tr[(\sigma_{\min} - \rho) 2W] |,
    \end{align}
    where the last equality holds due to the fact that $\tr(\sigma_{\min} - \rho)  = 0$.
    This implies that
    \begin{equation}
        \etr(\rho)\ge [\tr(\sigma_{\min} W) - \tr(\rho W)] \ge -\tr(\rho W).
    \end{equation}
    Since $\etr(\rho) \ge 0$, by definition of $w_k(\rho)$, we have $\etr(\rho) \ge w_k(\rho)$.
\end{proof}

\begin{theorem}\label{co:normineq} 
For any state $\rho$ and any $k$-network it holds that
\begin{equation}\label{eq:anyk}
    \betr(\rho) \ge \frac{\omega(\rho)}{6 n(n+k-2)} - \frac{\beta(\rho)}{6},
  \end{equation}
  where $\omega(\rho) = \sum_{ij} \Gamma_{ij}(\rho) - k \tr(\Gamma(\rho))$, $\beta(\rho) = \min \{r(1-\tau), 2\sqrt{1-\tau^2}\}$ with $\tau = \tr(\rho^2)$ and $r$ is the rank of $\rho$.
\end{theorem}
\begin{proof}
    For any network state $\sigma$ from the independent $k$-network, {\color{black}i.e., a network with only $k$-partite sources and without shared randomness}, we have
    \begin{align}
        \omega(\sigma) \le 0.
    \end{align}
    Notice that, the difference between its covariance matrix elements for different states can be bounded by
    \begin{align}
        \Gamma_{ij}(\rho) - \Gamma_{ij}(\sigma) &\le |\rho-\sigma|_{1} + |\rho\otimes\rho - \sigma\otimes\sigma|_{1} \le 3 |\rho-\sigma|_{1}.
    \end{align}
 
    Thus,
    \begin{align}
        \min_{\sigma\in \siq}|\rho -\sigma|_{1} &\ge \min_{\sigma\in \siq} \frac{\omega(\rho) - \omega(\sigma)}{3 n(n+k-2)}  \ge \frac{\omega(\rho)}{3 n(n+k-2)},
    \end{align}
    where the first inequality is from the fact that there are only $(n^2-n)+n(k-1)$ terms of $\Gamma_{ij}(\rho) - \Gamma_{ij}(\sigma)$ in $\omega(\rho)$.
    
    For a given mixed state and its decomposition $\rho = \sum_k p_t \rho_t$, we have
    \begin{align}
        \min_{\{\sigma_t\}\subset \siq} \sum_t p_t |\rho_t -\sigma_t|_{1} &\ge \min_{\{\sigma_t\}\subset  \siq} \frac{\sum_t p_t [\omega(\rho_t) - \omega(\sigma_t)]}{3 n(n+k-2)} \nonumber\\ 
        & \ge \frac{\sum_t p_t \omega(\rho_t)}{3 n(n+k-2)}\nonumber\\ 
        & = \frac{\omega(\rho) - [\sum_{ij} T_{ij} - k \tr(T)]}{3 n(n+k-2)}\nonumber\\ 
        & \ge \frac{\omega(\rho)}{3 n(n+k-2)} - \frac{\beta(\rho)}{3},
    \end{align}
    where the inequality on the third line is from the fact that $\Gamma(\rho) = \sum_t p_t \Gamma(\rho_t) + T$ and the last inequality holds since there are $n(n+k-2)$ terms of $T_{ij}$ in $[\sum_{ij} T_{ij} - k \tr(T)]$ and the absolute value of each of them is no more than $\beta(\rho)$. This completes the proof.
\end{proof}
Notice that
\begin{align}
    \omega(\rho) - \omega(\sigma) &= \tr[\tilde{M}_1 (\rho-\sigma)] - \tr[\tilde{M}_2 (\rho\otimes\rho-\sigma\otimes\sigma)] \le (\lambda(\tilde{M}_1) + 2\lambda(\tilde{M}_2)) |\rho-\sigma|_{1},
\end{align}
where  $\tilde{M}_1 = (\sum_{i} M_i)^2 -k\sum_i M_i^2$, $\tilde{M}_2 = (\sum_{i} M_i)^{\otimes 2} - k\sum_i M_i^{\otimes 2}$ whose maximal singular value is $\lambda(\tilde{M})$. Since the absolute eigenvalues of $M_i$'s are no more than $1$ by definition, $\lambda(\tilde{M}_i) \le n(n+k-2)$, which can result in a better lower bound for $\betr(\rho)$. More explicitly,
\begin{equation}\label{eq:anyk_tighter}
    \betr(\rho) \ge \frac{\omega(\rho)}{6\lambda(\tilde{M})} - \frac{\beta(\rho)}{6}.
\end{equation}
However, since we have to know the exact $M_i$'s for this estimation, it is no longer measurement-device independent.

\bibliographystyle{quantum}
\bibliography{network_states.bib}

\end{document}